\pdfoutput=1

\documentclass{scrartcl}

\usepackage{cmap}
\usepackage[utf8]{inputenc}
\usepackage[T1]{fontenc}
\usepackage{microtype}

\usepackage{amsmath}
\usepackage{amssymb}
\usepackage{mathtools}
\usepackage{mathrsfs} 
\usepackage{bm}
\usepackage{csquotes}
\usepackage{enumitem}
\usepackage{caption}
\usepackage{subcaption}

\usepackage{orcidlink}

\usepackage{amsthm}
\usepackage{thmtools}

\theoremstyle{plain}
\newtheorem{theorem}{Theorem}[section]

\newtheorem{proposition}[theorem]{Proposition}
\newtheorem{corollary}[theorem]{Corollary}
\newtheorem{fact}[theorem]{Fact}

\theoremstyle{definition}
\newtheorem{definition}[theorem]{Definition}
\newtheorem{example}[theorem]{Example}

\theoremstyle{remark}
\newtheorem{remark}[theorem]{Remark}
\newtheorem*{remark*}{Remark}

\usepackage{tikz}
\usetikzlibrary{automata, matrix, positioning, calc, decorations.pathreplacing, shapes.geometric}
\newlength{\edgelength}
\newcommand{\trans}[4]{%
  \begin{tikzpicture}[auto, shorten >=1pt, >=latex, baseline=(l.base), inner sep=0pt, outer xsep=0.3333em]
    \node (l) {\ensuremath{#1}};%
    \setlength{\edgelength}{\widthof{\scriptsize\ensuremath{#2/#3}}+0.5cm}%
    \node[base right=\edgelength of l] (r) {\ensuremath{#4}};%
    \path[->] (l.mid east) edge node[inner sep=0pt] {\scriptsize\ensuremath{#2/#3}} (r.mid west);%
  \end{tikzpicture}%
}

\usepackage{calc}
\usepackage{tabularx}
\newcommand{\problem}[3][]{%
  \par\vspace{0.125cm plus 0.1cm minus 0.05cm}\begin{tabularx}{\linewidth-2\parindent}{@{}lX}%
    \if\relax\detokenize{#1}\relax%
    \else%
    \textnormal{\textbf{Constant:}}&#1\\%
    \fi%
    \textnormal{\textbf{Input:}}&#2\\%
    \textnormal{\textbf{Question:}}&#3\\%
  \end{tabularx}\vspace{0.125cm plus 0.1cm minus 0.05cm}\par%
}

\newcommand{\changefont}[3]{\fontfamily{#1}\fontseries{#2}\fontshape{#3}\selectfont}
\newcommand*{\DecProblem}[1]{{\changefont{cmtt}{m}{sc}#1}}

\newcommand{\function}[3][]{%
  \par\vspace{0.125cm plus 0.1cm minus 0.05cm}\begin{tabularx}{\linewidth-2\parindent}{@{}lX}%
    \if\relax\detokenize{#1}\relax%
    \else%
    \textnormal{\textbf{Constant:}}&#1\\%
    \fi%
    \textnormal{\textbf{Input:}}&#2\\%
    \textnormal{\textbf{Output:}}&#3\\%
  \end{tabularx}\vspace{0.125cm plus 0.1cm minus 0.05cm}\par%
}

\newcommand*{\ComplexityClass}[1]{\textsc{#1}\xspace}

\newcommand*{\PSpace}{\ComplexityClass{PSpace}}

\newcommand*{\NP}{\ComplexityClass{NP}}
\newcommand*{\coNP}{\ComplexityClass{coNP}}

\newcommand*{\LogSpace}{\ComplexityClass{LogSpace}}

\usepackage{xspace}
\newcommand*{\SAut}{$\mathscr{S}$\kern-0.4ex-\allowbreak{}au\-to\-ma\-ton\xspace}
\newcommand*{\SAuta}{$\mathscr{S}$\kern-0.4ex-\allowbreak{}au\-to\-ma\-ta\xspace}
\newcommand*{\SIAut}{$\inverse{\mathscr{S}}$\kern-0.4ex-\allowbreak{}au\-to\-ma\-ton\xspace}
\newcommand*{\SIAuta}{$\inverse{\mathscr{S}}$\kern-0.4ex-\allowbreak{}au\-to\-ma\-ta\xspace}
\newcommand*{\GAut}{$\mathscr{G}$\kern-0.2ex-\allowbreak{}au\-to\-ma\-ton\xspace}
\newcommand*{\GAuta}{$\mathscr{G}$\kern-0.2ex-\allowbreak{}au\-to\-ma\-ta\xspace}

\newcommand*{\id}{\operatorname{id}}
\usepackage{bbold}
\newcommand{\idGrp}{\mathbb{1}}
\newcommand{\revbin}{\overleftarrow{\operatorname{bin}}}

\usepackage[affil-it]{authblk}

\author{Maximilian Kotowsky}
\affil{Insitut für Formale Methoden der Informatik (FMI)\\
  Universität Stuttgart\\
  Universitätsstraße 38\\
  70569 Stuttgart, Germany}
\author{Jan~Philipp~Wächter \orcidlink{0000-0002-7801-6569}\thanks{The second author was funded by the Deutsche Forschungsgemeinschaft (DFG, German Research Foundation) – 492814705.}}
\affil{Universität des Saarlandes\\
  Fachrichtung Mathematik\\
  Campus E2 4\\
  66123 Saarbrücken, Germany}

\usepackage{hyperref}
\hypersetup{
  colorlinks,
  linkcolor={red!50!black},
  citecolor={blue!50!black},
  urlcolor={blue!80!black}
}

\title{The Word Problem for Finitary Automaton Groups}

\begin{document}
  \maketitle
  
  \begin{abstract}
    A finitary automaton group is a group generated by an invertible, deterministic finite-state letter-to-letter transducer whose only cycles are self-loops at an identity state. We show that, for this presentation of finite groups, the uniform word problem is \coNP-complete. Here, the input consists of a finitary automaton together with a finite state sequence and the question is whether the sequence acts trivially on all input words. Additionally, we also show that the respective compressed word problem, where the state sequence is given as a straight-line program, is \PSpace-complete. In both cases, we give a direct reduction from the satisfiability problem for (quantified) boolean formulae and we further show that the problems remain complete for their respective classes if we restrict the input alphabet of the automata to a binary one.
    \newline
    \textbf{Keywords.} Automaton Group, Word Problem, Finitary, Activity.
  \end{abstract}
  
  \begin{section}{Introduction}
    There are many connections between groups and automata (see e.\,g.\ \cite{silva2012groups}). In this article, we are mostly concerned with automaton groups, where the term automaton refers to an invertible, deterministic finite-state letter-to-letter transducer. In such an automaton, every state $q$ induces a function mapping an input word $u$ to the output word obtained by starting in $q$ and following the path labeled by $u$ in the input. Since the automaton is invertible, every such function is a bijection and the closure under composition of these functions (and their inverses) forms a group. This is the group generated by the automaton and any group arising in this way is an automaton group. Not every group is an automaton group but the class of automaton groups contains some very interesting examples (see e.\,g.\ \cite{bartholdi2010groups}). Probably the most famous one is Grigochuk's group, which -- among other interesting properties -- was the historically first group of intermediate growth (i.\,e.\ the numbers of elements that can be written as a word of length at most $n$ over the generators grow slower than any exponential function but faster than any polynomial; see \cite{grigorchuk2008groups} for an introduction to this topic).
    
    These interesting examples also led to an investigation of the algorithmic properties of automaton groups, where the presentation using automata is an alternative to the classical one using (typically finitely many) generators and relations. It turns out that this presentation is still quite powerful as many decision problems remain undecidable. For example, it is known that there is an automaton group with an undecidable conjugacy problem \cite{sunic2012conjugacy} (given two group elements, check whether they are conjugate) and one with an undecidable order problem \cite{gillibert2018automaton, bartholdi2017word} (given a group element, check whether it has finite order). Decidability of the finiteness problem for automaton groups (given an automaton, check whether its generated group is finite) is still an open problem but the corresponding problem for semigroups has been shown to be undecidable \cite{gillibert2014finiteness}.
    
    The word problem (given a group element, check whether it is the neutral element), however, seems to have a special role for automaton groups. It is well known to be decidable and a guess and check approach also yields that the problem can be solved in non-deterministic linear space, even in the uniform case (where the generating automaton is also part of the input) \cite{steinberg2015some, dangeli2017complexity}. Regarding lower bounds, Armin Weiß and the second author proved that there is an automaton group with a \PSpace-complete word problem \cite{waechter2022automaton}.
    
    In this work, we will apply similar ideas to investigate the complexity of the word problem for the lowest level of the activity hierarchy for automaton groups introduced by Sidki \cite{sidki2000automorphisms}. This hierarchy classifies automaton groups based on the structure of the cycles in the generating automaton. At the lowest level, which belongs to the class of finitary automata and finitary automaton groups, the only cycles are the self-loops at an identity state (i.\,e.\ a state where the output word is always the same as the input word). It turns out that this class coincides with the class of (all) finite groups.
    
    On the next level, the class of bounded automata and bounded automaton groups, every path in the automaton may contain at most one cycle (not counting the self-loops at a possible identity state). This class still seems \enquote{finite enough} for many problems to be decidable. For example, the finiteness problem \cite{waechter2021orbits} as well as the order problem \cite{bondarenko2013conjugacy} are decidable and there are positive results on the conjugacy problem \cite{bondarenko2013conjugacy}; the word problem of a bounded automaton group can be solved in deterministic logarithmic space \cite{nekrashevych2005self, bartholdi2019groups} and its complement is an ET0L language \cite{bishop2019bounded}.
    
    We will be interested in the finitary level. As we have discussed, studying the word problem of these groups is the same as studying the word problem of arbitrary finite groups. It is well known that a group is finite if and only if its word problem (i.\,e.\ the formal language of words over the generators representing the neutral element) is regular. While this does not settle the precise complexity for the individual groups entirely, we will approach this setting from a different perspective. We will consider the uniform word problem, where the group is part of the input in a suitable presentation. Typical such presentations include, for example, the classical one with generators and relations, Caley graphs and tables or presenting the elements as matrices or permutations (where the representation as permutations may be considered a special case of the representation as matrices). For Cayley tables, the problem can be solved in deterministic logarithmic space (by iterated lookups in the table) and the same is true for matrix representations \cite{1lipton1977word}. Since the word problem of every non-solvable finite group is $\ComplexityClass{NC}^1$-complete \cite{barrington89boundedWidth}, we immediately get a lower bound for any group representation. For permutations, there are also lower bound results regarding deterministic logarithmic space \cite{cook1987problems}.
    
    Our presentation of choice is that of using an automaton (in the way described above). Here, we will show that the uniform word problem is \coNP-complete by giving a direct reduction from the satisfiability problem for boolean formulae. We even show that the problem remains \coNP-complete if we restrict the possible input automata to ones with a binary alphabet.
    Then, we will show that the uniform compressed word problem, where the input state sequence is not given directly but only compressed in the form of a context-free grammar (or, more precisely, a straight-line program), is \PSpace-complete and, thus, exponentially harder (under common complexity theoretic assumptions). This reflects a similar (provable) exponential gap in the general case \cite{waechter2022automaton}. We prove this latter result by giving a direct reduction from the satisfiability problem for quantified boolean formulae but also by approximating the word problem of Grigorchuk's group (which is known to be \PSpace-complete \cite{bartholdi2019groups}) using finitary automata. The latter approach is less direct but shows that also the uniform compressed word problem for finitary automaton groups remains \PSpace-complete if we restrict the input to automata with binary alphabet.
    
    The approach of simulating logical formulae in automata is similar to the techniques used in \cite{waechter2022automaton} and we hope that the general idea can be extended to further settings, for example, to obtain lower bound results for further levels of the activity hierarchy. The underlying idea is to use certain commutators for simulating logical conjunctions. This is often attributed to Barrington, who used this approach to show the above-mentioned result on the $\ComplexityClass{NC}^1$-completeness of the word problem of non-solvable finite groups \cite{barrington89boundedWidth} (see \cite{bartholdi2019groups} for more results in that direction). However, there are also similar ideas predating Barrington \cite{makanin1984decidability, malcev1962equation, maurer1965property, krohn1966realizing}.
  \end{section}

  \begin{section}{Preliminaries}
    \paragraph*{Logic.}
    For this paper, we will require some basic knowledge about propositional and first-order logic. We use $\bot$ to denote a \emph{false} truth value and $\top$ to denote the truth value \emph{true}. We let $\mathbb{B} = \{ \bot, \top \}$ and may evaluate the truth value $\mathcal{A}(\varphi)$ of a formula $\varphi$ over the \emph{variables} $\mathbb{X}$ under an \emph{assignment} $\mathcal{A}: \mathbb{X} \to \mathbb{B}$ in the usual way. If this evaluates to $\top$, we say that $\mathcal{A}$ \emph{satisfies} $\varphi$ and $\varphi$ is \emph{satisfiable} if it is satisfied by some assignment. A \emph{literal} is either a variable $x$ or the negation $\lnot x$ of a variable. In the first case, the literal is \emph{positive} and, in the second case, it is \emph{negative}. A \emph{clause} is a disjunction $\bigvee_{i = 1}^n L_i$ of literals $L_i$. A conjunction $\bigwedge_{k = 1}^K C_k$ of clauses $C_k$ is a formula in \emph{conjunctive normal form}. If all the clauses contain exactly $3$ distinct literals, we say that the formula is in \emph{$3$-conjunctive normal form}.
    
    \paragraph*{Complexity.}
    We need some notions from complexity theory for this paper. However, we will not go into details about complexity theory and refer the reader to standard textbooks on the topic (such as \cite{papadimitriou97computational}) instead. Regarding complexity classes, we need the class \coNP which contains all problems whose complements can be solved in non-deterministic polynomial time (i.\,e.\ are in \NP) and the class \PSpace of problems solvable in polynomial space (where it does not matter whether we consider deterministic or non-deterministic algorithms by Savitch's theorem \cite[Theorem~7.5]{papadimitriou97computational}).
    We point out that a problem is in \PSpace if and only if its complement is.
    Additionally, we need \LogSpace-computable functions (where \LogSpace refers to deterministic logarithmic space). When it comes to reductions, we will exclusively work with \emph{many-one \LogSpace-reductions}. Formally, such a reduction from a problem $A$ to a problem $B$ is a \LogSpace-computable function $f$ mapping instances of $A$ to instances of $B$ such that positive instances are mapped to positive instances and negative instances are mapped to negative ones. A problem $A$ is \emph{$\mathcal{C}$-hard} for some complexity class $\mathcal{C}$ if any problem $C \in \mathcal{C}$ can be reduced to $A$ (using a many-one \LogSpace-reduction). Typically, this is done by reducing a problem which is already known to be $\mathcal{C}$-hard to $A$ (as many-one $\LogSpace$-reductions are closed under composition, see e.\,g.\ \cite[Proposition~8.2]{papadimitriou97computational}). If a $\mathcal{C}$-hard problem is also contained in $\mathcal{C}$, it is \emph{$\mathcal{C}$-complete}.
    
    \paragraph*{Words and Group Operations.}
    An alphabet is a non-empty, finite set $\Sigma$. A finite sequence $w = a_1 \dots a_\ell$ of elements $a_1, \dots, a_\ell \in \Sigma$ is a \emph{word} and its \emph{length} is $|w| = \ell$. The unique word of length $0$ is denoted by $\varepsilon$ and the set of all words over $\Sigma$ is $\Sigma^*$, which forms a monoid whose operation is the concatenation of words (and whose neutral element is $\varepsilon$).
    We use $\Sigma^\ell$ to denote the set of words of length exactly $\ell$ and natural variations of this notation (such as writing $\Sigma^{< \ell}$ for the set of words of length (strictly) smaller than $\ell$).
    
    We will often work with words in the context of generating a group. In this case, we assume that, for an alphabet $Q$, we have a disjoint copy $Q^{-1} = \{ q^{-1} \mid q \in Q \}$ of formal inverse letters. For the set of words over such positive and negative letters, we write $Q^{\pm *} = (Q \cup Q^{-1})^*$ and we may extend the notation $q^{-1}$ to words by letting $(q_1 \dots q_\ell)^{-1} = q_\ell^{-1} \dots q_1^{-1}$ where we additionally use the convention $(q^{-1})^{-1} = q$. We say a group $G$ is \emph{generated} by $Q$ if there is a monoid homomorphism $\pi: Q^{\pm *} \to G$ with $\pi(q^{-1}) = \pi(q)^{-1}$. In this context, we write $\bm{p} = \bm{q}$ \emph{in $G$} for $\pi(\bm{p}) = \pi(\bm{q})$ (where $\bm{p}, \bm{q} \in Q^{\pm *}$) and also $\bm{p} = g$ \emph{in $G$} if $\pi(\bm{p}) = g$. So, for example, we write $\bm{p} = \idGrp$ in $G$ if $\pi( \bm{p} )$ is the neutral element of the group $G$, which we usually denote by $\idGrp$.
    
    In addition to taking the inverse, we lift further group operations to words. In analogy to the \emph{conjugation} $g^k = k^{-1} g k$ of some group element $g \in G$ by another one $k \in G$, we also write $\bm{q}^{\bm{p}}$ for the word $\bm{q}^{\bm{p}} = \bm{p}^{-1} \bm{q} \bm{p}$ (where $\bm{p}, \bm{q} \in Q^{\pm *}$). Note that this notation is compatible with the conjugation as we have $\pi(\bm{q}^{\bm{p}}) = \pi(\bm{q})^{\pi(\bm{p})}$. We also do the same for the \emph{commutator} $[h, g] = h^{-1} g^{-1} hg$ of two group elements $g, h \in G$ and write $[\bm{q}, \bm{p}]$ for the word $[\bm{q}, \bm{p}] = \bm{q}^{-1} \bm{p}^{-1} \bm{q} \bm{p}$. Again, this is compatible with the projection $\pi$: $\pi\left( [\bm{q}, \bm{p}] \right) = [ \pi(\bm{q}), \pi(\bm{p}) ]$.

    \paragraph*{Automata and Automaton Groups.}
    In the context of this paper, an automaton is a finite state letter-to-letter transducer. Formally, an \emph{automaton} $\mathcal{T}$ is a triple $(Q, \Sigma, \delta)$ where $Q$ is a finite, non-empty set of \emph{states}, $\Sigma$ is the (input and output) \emph{alphabet} of $\mathcal{T}$ and $\delta \subseteq Q \times \Sigma \times \Sigma \times Q$ is the \emph{transition} relation. In this context, we usually write $\trans{p}{a}{b}{q}$ for the tuple $(p, a, b, q) \in Q \times \Sigma \times \Sigma \times Q$. This is a transition \emph{starting} in $p$, \emph{ending} in $q$ with \emph{input} $a$ and \emph{output} $b$.
    
    An automaton $\mathcal{T} = (Q, \Sigma, \delta)$ is \emph{deterministic} and \emph{complete} if we have $d_{p, a} = | \{ \trans{p}{a}{b}{q} \mid b \in \Sigma, q \in Q \} | = 1$ for all $p \in Q$ and $a \in \Sigma$. It is additionally \emph{invertible} if we also have $d'_{p, b} = | \{ \trans{p}{a}{b}{q} \mid a \in \Sigma, q \in Q \} | = 1$ for all $p \in Q$ and $b \in \Sigma$. We will call a deterministic, complete and invertible automaton a \emph{\GAut}.
    
    \begin{figure}%
      \centering%
      \begin{subfigure}[t]{0.33\linewidth}%
        \centering%
        \begin{tikzpicture}[baseline=(m-2-1.base)]
          \matrix[matrix of math nodes, text height=1.25ex, text depth=0.25ex] (m) {
            & a & \\
            p & & q \\
            & b & \\
          };
          \foreach \i in {1} {
            \draw[->] let
              \n1 = {int(2+\i)}
            in
              (m-2-\i) -> (m-2-\n1);
            \draw[->] let
              \n1 = {int(1+\i)}
            in
              (m-1-\n1) -> (m-3-\n1);
          };
        \end{tikzpicture}%
        \caption{Cross diagrams}\label{fig:singleCrossDiagram}%
      \end{subfigure}\hfill%
      \begin{subfigure}[t]{0.33\linewidth}%
        \centering%
        \begin{tikzpicture}[baseline=(m-2-1.base)]
          \matrix[matrix of math nodes, text height=1.25ex, text depth=0.25ex] (m) {
            & b & \\
            p^{-1} & & q^{-1} \\
            & a & \\
          };
          \foreach \i in {1} {
            \draw[->] let
             \n1 = {int(2+\i)}
            in
              (m-2-\i) -> (m-2-\n1);
            \draw[->] let
              \n1 = {int(1+\i)}
            in
             (m-1-\n1) -> (m-3-\n1);
          };
        \end{tikzpicture}%
        \caption{Inverse cross diagrams}\label{fig:inverseCrossDiagram}%
      \end{subfigure}\hfill%
      \begin{subfigure}[t]{0.33\linewidth}%
        \centering%
        \begin{tikzpicture}[baseline=(m-2-1.base)]
          \matrix[matrix of math nodes, text height=1.25ex, text depth=0.25ex, ampersand replacement=\&] (m) {
                   \& u \& \\
            \bm{p} \&   \& \bm{q} \\
                   \& v \& \\
          };
          \foreach \i in {1} {
            \draw[->] let
              \n1 = {int(2+\i)}
            in
              (m-2-\i) -> (m-2-\n1);
            \draw[->] let
              \n1 = {int(1+\i)}
            in
              (m-1-\n1) -> (m-3-\n1);
          };
        \end{tikzpicture}%
        \caption{Abbreviated cross diagram}\label{fig:abbreviatedCrossDiagram}%
      \end{subfigure}\newline
      \hfill\begin{subfigure}[t]{0.6\linewidth}%
        \centering%
        \begin{tikzpicture}[baseline=(m-4-1.east)]
          \matrix[matrix of math nodes, text height=1.25ex, text depth=0.25ex, ampersand replacement=\&] (m) {
                     \& a_{0, 1}     \&          \& \dots \&              \& a_{0, m}     \&     \\
            q_{1, 0} \&              \& q_{1, 1} \& \dots \& q_{1, m - 1} \&              \& q_{1, m} \\
                     \& a_{1, 1}     \&          \&       \&              \& a_{1, m}     \&     \\
              \vdots \& \vdots       \&          \&       \&              \& \vdots       \& \vdots \\
                     \& a_{n - 1, 1} \&          \&       \&              \& a_{n - 1, m} \&     \\
            q_{n, 0} \&              \& q_{n, 1} \& \dots \& q_{n, m - 1} \&              \& q_{n, m} \\
                     \& a_{n, 1}     \&          \& \dots \&              \& a_{n, m}     \&     \\
          };
          \foreach \j in {1, 5} {
            \foreach \i in {1, 5} {
              \draw[->] let
                \n1 = {int(2+\i)},
                \n2 = {int(1+\j)}
              in
                (m-\n2-\i) -> (m-\n2-\n1);
              \draw[->] let
                \n1 = {int(1+\i)},
                \n2 = {int(2+\j)}
              in
                (m-\j-\n1) -> (m-\n2-\n1);
            };
          };
        \end{tikzpicture}%
        \caption{Multiple crosses combined in one diagram}\label{fig:combinedCrossDiagram}%
      \end{subfigure}\hfill%
      \caption{Single, inverted, combined and abbreviated cross diagrams}
    \end{figure}
    Another way of indicating that we have a transition $\trans{p}{a}{b}{q} \in \delta$ is to use the cross diagram in \autoref{fig:singleCrossDiagram}.
    Multiple cross diagrams may be combined into a larger one. For example, the cross diagram in \autoref{fig:combinedCrossDiagram} indicates that we have $\trans{q_{i, j - 1}}{a_{i - 1, j}}{a_{i, j}}{q_{i, j}} \in \delta$ for all $1 \leq i \leq n$ and $1 \leq j \leq m$.
    Typically, we will omit unnecessary intermediate states if we do not need to name them. Additionally, we also allow abbreviations in the form of words (instead of only single letters) in the input and output and state sequences (i.\,e.\ words over $Q$) on the left and the right. Note, however, that here the right-most state of the sequence is considered to be the first state,\footnote{This makes sense as we will later on define a \emph{left} action of the states on the words over $\Sigma$.} which results in the abbreviated cross diagram in \autoref{fig:abbreviatedCrossDiagram} for $\bm{p} = q_{n,0} \dots q_{1, 0}$, $u = a_{0, 1} \dots a_{0, m}$, $v = a_{n, 1} \dots a_{n, m}$ and $\bm{q} = q_{n, m} \dots q_{1, m}$.
    
    For a deterministic and complete automaton $\mathcal{T} = (Q, \Sigma, \delta)$, there exists exactly one cross diagram of the form in \autoref{fig:abbreviatedCrossDiagram} for every $\bm{p} \in Q^*$ and $u \in \Sigma^*$. If $\mathcal{T}$ is additionally invertible (i.\,e.\ it is a \GAut), we define that we have the cross diagram in \autoref{fig:inverseCrossDiagram} for $p, q \in Q$ and $a, b \in \Sigma$ whenever we have the cross diagram from \autoref{fig:singleCrossDiagram}. Note that we have flipped the cross diagram along its horizontal axis and inverted the states. In this case, the cross diagram in \autoref{fig:abbreviatedCrossDiagram} uniquely exists for all $\bm{p} \in Q^{\pm *}$ and all $u \in \Sigma^*$ (although we now also allow states from $Q^{-1}$).
    
    This allows us to define a left action of $Q^{\pm *}$ on $\Sigma^*$ where the action of $\bm{q} \in Q^{\pm *}$ on a word $u \in \Sigma^*$ is given by $\bm{q} \circ u = v$ where $v$ is uniquely obtained from the cross diagram \autoref{fig:abbreviatedCrossDiagram} (the empty state sequence acts as the identity on all words by convention). The reader may verify that we indeed have $\bm{q}^{-1} \bm{q} \circ u = u = \bm{q} \bm{q}^{-1} \circ u$ with our definition of inverting cross diagrams.
    
    For two state sequences $\bm{p}, \bm{q} \in Q^{\pm *}$ of a \GAut $\mathcal{T} = (Q, \Sigma, \delta)$, we define the relation
    \[
      \bm{p} =_\mathcal{T} \bm{q} \iff \forall u \in \Sigma^*: \bm{p} \circ u = \bm{q} \circ u \text{.}
    \]
    It turns out that this relation is a congruence, which allows us to consider the monoid $Q^{\pm *} / {=_\mathcal{T}}$ formed by its classes. In fact, this monoid has a group structure (where the class of $\bm{q}^{-1}$ is the inverse of the class of $\bm{q}$) and this is the \emph{group generated} by $\mathcal{T}$. Any group generated by some \GAut is called an \emph{automaton group}.
    
    \paragraph*{The Dual Action.}
    The \GAut $\mathcal{T} = (Q, \Sigma, \delta)$ does not only induce the left action $\bm{p} \circ u$ for every $\bm{p} \in Q^{\pm*}$ and every $u \in \Sigma^*$. We may also define a right action of $\Sigma^*$ on $Q^{\pm *}$. Since the cross diagram in \autoref{fig:abbreviatedCrossDiagram} is unique for all $\bm{p} \in Q^{\pm *}$ and $u \in \Sigma^*$, we may let $\bm{p} \cdot u = \bm{q}$. The reader may observe that this is indeed a right action (from the way cross diagrams work) and this is called the \emph{dual action} (of $\mathcal{T}$). It is compatible with the relation $=_{\mathcal{T}}$ (i.\,e.\ we have $\bm{p} =_{\mathcal{T}} \bm{q} \implies \bm{p} \cdot u =_{\mathcal{T}} \bm{q} \cdot u$ for all $\bm{p}, \bm{q} \in Q^{\pm*}$ and all $u \in \Sigma^*$), which means that the dual action is actually a right action of $\Sigma^*$ on $\mathscr{G}(\mathcal{T})$. Furthermore, there is an interaction between the action $\bm{q} \circ u$ and the dual action $\bm{q} \cdot u$. In particular, we have
    \[
      \bm{q} \circ uv = (\bm{q} \circ u) (\bm{q} \cdot u \circ v)
    \]
    for all $\bm{q} \in Q^{\pm*}$ and all $u, v \in \Sigma^*$.
    
    \paragraph*{Graphical Depiction and Finitary Automata.}
    We use the common graphical depiction of automata, which results in a $\Sigma \times \Sigma$-labeled finite directed graph (see, e.\,g.\ \autoref{fig:grigorchuk}). If this graph does not have any cycles except for the $a / a$ labeled self-loops at an identity state,\footnote{Note that any complete finite automaton must contain a cycle and that, thus, every finitary \GAut has an identity state.} we say that the automaton is \emph{finitary}. The \emph{depth} of a finitary \GAut is the minimal number $d$ such that, after reading at least $d$ many letters, we are always in the identity state (regardless of where we started). A group generated by a finitary \GAut is a \emph{finitary} automaton group. Since, with a finitary \GAut, a state sequence may only act non-trivially on the first $d$ letters (where $d$ is the depth of the generating automaton), a finitary automaton group is necessarily finite. On the other hand, any finite group $G$ is generated by the finitary automaton $(G, G, \delta)$ with $\delta = \{ \trans{g}{h}{gh}{\idGrp} \mid g, h \in G \}$. Thus, studying finitary automaton groups is the same as studying finite groups but we are interested in a certain way of presenting these groups.
    
    \paragraph*{Contracting Automata.}
    The next step after finitary automaton groups in Sidki's activity hierarchy \cite{sidki2000automorphisms} is the class of bounded automaton groups. We will not give precise definitions here (see e.\,g.\ \cite{waechter2024word} for those and more information) but instead only mention that a \GAut is \emph{bounded} if every run has at most one cycle that does not consist of only $a/a$ self-loops at an identity state.
    
    Instead of looking at bounded automaton groups, we will need the (more general) notion of contracting automata. For this, we first define the \emph{word norm} induced by a \GAut $\mathcal{T} = (Q, \Sigma, \delta)$. It is the function $\| \cdot \| : Q^{\pm*} \to \mathbb{N}$ with
    \[
      \| \bm{q} \| = \min \{ | \bm{q}' | \mid \bm{q}' =_{\mathcal{T}} \bm{q}, \bm{q}' \in Q^{\pm*} \} \text{.}
    \]
    Clearly, $\bm{p} =_{\mathcal{T}} \bm{q}$ implies $\| \bm{p} \| = \| \bm{q} \|$ and we may, thus, consider the word norm as a function $\mathscr{G}(\mathcal{T}) \to \mathbb{N}$. The word norm is \emph{contracting} if there are constants $\lambda > 1$, $\nu \geq 0$ and $\ell \in \mathbb{N}$ with
    \[
      \| \bm{q} \cdot u \| \leq \frac{1}{\lambda} \| \bm{q} \| + \nu
    \]
    for all $\bm{q} \in Q^{\pm*}$ and $u \in \Sigma^\ell$. A \GAut is \emph{contracting} if its induced word norm is. The values of $\lambda$, $\nu$ and $\ell$ are called the \emph{contraction constants} of the automaton.
    
    The idea behind the notion of a contracting automaton is that group elements get \emph{shorter} after reading a block of $\ell$ letters until we eventually reach an element from a finite set (again, for more information on what this means precisely, we refer the reader for example to \cite{waechter2024word}).
    
    \begin{figure}\centering
      \begin{tikzpicture}[auto, shorten >=1pt, >=latex, baseline=(b.base)]
        \node[state] (b) {$b$};
        \node[state, above right=of b] (a) {$a$};
        \node[state, below right=of b] (d) {$d$};
        \node[state, below right=of a] (c) {$c$};
        \node[state, right=1cm of c] (id) {$\id$};

        \draw[->] (a) edge[bend left] node[align=center] {$0/1$\\$1/0$} (id)
                  (b) edge node {$0/0$} (a)
                  (b) edge node[swap] {$1/1$} (c)
                  (c) edge node[swap] {$0/0$} (a)
                  (c) edge node {$1/1$} (d)
                  (d) edge[bend right] node[swap] {$0/0$} (id)
                  (d) edge node {$1/1$} (b)
                  (id) edge[loop right] node[align=center] {$0/0$\\$1/1$} (id)
        ;
      \end{tikzpicture}
      \caption{Automaton generating Grigorchuk's group}\label{fig:grigorchuk}
    \end{figure}
    \begin{figure}\centering
      \begin{tikzpicture}[baseline=(m-2-1.base)]
        \matrix[matrix of nodes, text height=1.25ex, text depth=0pt, ampersand replacement=\&,
        ] (m) {
          \& \tiny $0$ or $1$ \& \\
          $(\star)$ \& \& $(\star, a \text{ or } \textnormal{id})$ \\
          \& \tiny $0$ or $1$ \& \\
          $a$ \& \& \textcolor{gray}{$\id$} \\
          \& \tiny $1$ or $0$ \& \\
          $\star$ \& \& $\star, a \text{ or } \textnormal{id}$ \\
          \& \tiny $1$ or $0$ \& \\
          \raisebox{-0.5ex}{\vdots} \& \raisebox{-0.5ex}{\vdots} \& \raisebox{-0.5ex}{\vdots} \\
          \& \tiny $1$ or $0$ \& \\
          $a$ \& \& \textcolor{gray}{$\id$} \\
          \& \tiny $0$ or $1$ \& \\
          $(\star)$ \& \& $(\star, a \text{ or } \textnormal{id})$ \\
          \& \tiny $0$ or $1$ \& \\
        };
        \draw[decorate, decoration={brace}] (m-12-1.south west) -- node[left, xshift=-2pt] {$\bm{p}$} (m-2-1.north west);
        \foreach \j in {1,3,5,9,11} {
          \foreach \i in {1} {
            \draw[->] let
              \n1 = {int(2+\i)},
              \n2 = {int(1+\j)}
            in
              (m-\n2-\i) -> (m-\n2-\n1);
            \draw[->] let
              \n1 = {int(1+\i)},
              \n2 = {int(2+\j)}
            in
              (m-\j-\n1) -> (m-\n2-\n1);
          };
        };
        \draw[decorate, decoration={brace}] (m-2-3.north east) -- node[right, yshift=-2.5ex, xshift=2pt] {$\bm{p} \cdot z$ with $\begin{aligned}[t]
            \| \bm{p} \cdot z \| &\leq |\bm{p}| {}- |\bm{p}|_a \\
            &\leq \frac{|\bm{p}| - 1}{2}
          \end{aligned}$} (m-12-3.south east);
      \end{tikzpicture}
      \caption{Illustration of the main idea for proving that Grigorchuk's group is contracting}\label{fig:grigorchukContracting}
    \end{figure}
    \begin{example}[Grigorchuk's Group]\label{ex:grigorchuk}
      One of the most well-known automaton groups (and probably the most famous one) is \emph{Grigorchuk's group}. It is generated by the \GAut $\mathcal{G}$ with binary alphabet depicted in \autoref{fig:grigorchuk}. We refer the reader for more details on Grigorchuk's group to \cite{grigorchuk2008groups}, \cite{nekrashevych2005self} or \cite{waechter2024word} but we will give the main idea of why it is contracting (with constants $\lambda = {2}$ and $\nu = \ell = 1$).

      The action of $a$ on an input word is to toggle the first letter between $0$ and $1$. This implies $a^2 = \idGrp$ in $\mathscr{G}(\mathcal{G})$. With the proper case distinctions (on the length of a prefix only consisting of symbols $1$ modulo $3$),\footnote{Again, see \cite{grigorchuk2008groups} or \cite{waechter2024word} for a full proof.} one can prove that we also have $b^2 = c^2 = d^3 = \idGrp$ in $\mathscr{G}(\mathcal{G})$; in particular, every state is equal to its own inverse in the generated group. Finally, the product of two distinct elements $x, y \in \{ b, c, d \}$ is the third one from this set (which may be seen using a similar case distinction).
      This immediately yields a length reducing rewriting system and shows that any state sequence $\bm{q} \in Q^{\pm*}$ is equal to a state sequence $\bm{p} \in Q^*$ of the form
      \[
        (\star) a \star a \dots \star a (\star)
      \]
      where $\star$ are (possibly different) elements from $\{ b, c, d \}$
      (i.\,e.\ it is alternating between an $a$ and a letter from $\{ b, c, d \}$ and may start and end with either). We may (without loss of generality) assume that $\| \bm{q} \| = \| \bm{p} \| = | \bm{p} |$ for this state sequence and obtain
      \[
        | \bm{p} |_a \geq \frac{|\bm{p}| - 1}{2} \geq \frac{|\bm{p}|}{2} - 1
      \]
      where $|\bm{p}|_a$ is the result of counting how often the state $a$ appears in $\bm{p}$.
      
      Now, since we have $a \cdot 0 = a \cdot 1 = \id$, we obtain\footnote{Compare to \autoref{fig:grigorchukContracting}.}
      \[
        \| \bm{q} \cdot z \| = \| \bm{p} \cdot z \| \leq |\bm{p}| - |\bm{p}|_a \leq |\bm{p}| - \frac{|\bm{p}|}{2} + 1 = \frac{1}{2} |\bm{p}| + 1 = \frac{1}{2} \| \bm{q} \| + 1
      \]
      for all $z \in \{ 0, 1 \}$.
    \end{example}
    \begin{remark}
      There is a different way to see that the automaton generating Grigorchuk's group from \autoref{fig:grigorchuk} is contracting (although the constants are not immediately obvious from this approach): it is bounded and every bounded automaton is contracting \cite{bondarenko2003postCritically} (alternatively, see \cite{nekrashevych2005self} or \cite{waechter2024word}).
    \end{remark}
    
    The main point of contracting automata is that they allow us to distinguish different elements of the generated group by a witness of logarithmic length.
    \begin{fact}\label{fct:contractingLogarithmicWitness}
      For every contracting \GAut $\mathcal{T} = (Q, \Sigma, \delta)$, there are constants $A$ and $B$ such that, for every $\bm{p} \in Q^{\pm *}$ with $\bm{p} \neq \idGrp$ in $\mathscr{G}(\mathcal{T})$, there is a witness $w \in \Sigma^*$ of length $|w| \leq A \log |\bm{p}| + B$ with $\bm{p} \circ w \neq w$ (where $\log n$ denotes the logarithm of $n$ for basis $2$).
    \end{fact}
    \begin{proof}
      Let $\lambda, \nu$ and $\ell$ be the contraction constants for $\mathcal{T}$ and choose $A = \frac{\ell}{\log \lambda}$.
      Also let
      \[
        L = 1 + \nu \frac{1}{1 - \lambda^{-1}}
      \]
      and observe that there are only finitely many state sequences $\bm{q} \in Q^{\pm *}$ of length $|\bm{q}| \leq L$. Therefore, there is some constant $F$ such that, for every such $\bm{q}$ with $\bm{q} \neq \idGrp$ in $\mathscr{G}(\mathcal{T})$, there is a witness $w$ of length $|w| \leq F$ with $\bm{q} \circ w \neq w$. In fact, we may extend this statement to all $\bm{q}$ that are equivalent under $=_{\mathcal{T}}$ to a state sequence of length at most $L$, i.\,e.\ to all $\bm{q} \in Q^{\pm *}$ with $\| \bm{q} \| \leq L$. For this choice of $F$, let $B = \ell + F$.
      
      Now, fix some $\bm{p} \in Q^{\pm*}$ with $\bm{p} \neq \idGrp$ in $\mathscr{G}(\mathcal{T})$. We first use the contraction property to show by induction that for every word $u \in \Sigma^*$ of length $|u| = \ell s$ (with $s \geq 0$), we have
      \[
        \| \bm{p} \cdot u \| \leq \lambda^{-s} \| \bm{p} \| + \nu \sum_{i = 0}^{s - 1} \lambda^{-i} \text{.}
      \]
      This is immediate for $s = 0$. For the inductive step from $s$ to $s + 1$, we may write $u = u_0 u_1$ for $|u_0| = \ell s$ and $|u_1| = \ell$ and obtain (from the contraction property and induction):
      \begin{align*}
        \| \bm{p} \cdot u_0 u_1 \| &\leq \lambda^{-1} \| \bm{p} \cdot u_0 \| + \nu\\
        &\leq \lambda^{-1} \left ( \lambda^{-s} \| \bm{p} \| + \nu \sum_{i = 0}^{s - 1} \lambda^{-i} \right) + \nu\\
        &= \lambda^{-(s + 1)} \| \bm{p} \| + \nu \left( \sum_{i=0}^{s - 1} \lambda^{-(i + 1)} \right) + \nu\\
        &= \lambda^{-(s + 1)} \| \bm{p} \| + \nu \left( \sum_{i=1}^{s} \lambda^{-i)} \right) + \nu
        = \lambda^{-(s + 1)} \| \bm{p} \| + \nu \sum_{i=0}^{(s + 1) - 1} \lambda^{-i}
      \end{align*}
      
      Now, let $N$ be a natural number with
      \[
        \log_\lambda \| \bm{p} \| = \frac{\log \| \bm{p} \|}{\log \lambda} \leq N \leq \frac{\log \| \bm{p} \|}{\log \lambda} + 1
      \]
      (where $\log_\lambda n$ denotes the logarithm of $n$ for basis $\lambda$). We will show by induction that, for all $N \geq s \geq 0$ and $u \in \Sigma^*$ of length $|u| = \ell s$, there is some witness $w \in \Sigma^*$ of length $|w| \leq \ell (N - s) + F$ with $\bm{p} \cdot u \circ w \neq w$ if $\bm{p} \cdot u \neq \idGrp$ in $\mathscr{G}(\mathcal{T})$.
      
      For $s = N$, we obtain from the above:
      \begin{align*}
        \| \bm{p} \cdot u \| &\leq \lambda^{-N} \| \bm{p} \| + \nu \sum_{i = 0}^{N - 1} \lambda^{-i} \leq \left( \lambda^{\log_\lambda \| \bm{p} \|} \right)^{-1} \| \bm{p} \| + \nu \sum_{i \geq 0} \lambda^{-i} \\
        &\leq 1 + \nu \frac{1}{1 - \lambda^{-1}} = L
      \end{align*}
      Thus, there is indeed a witness $w \in \Sigma^*$ of length $|w| \leq F$ with $\bm{p} \cdot u \circ w \neq w$ (by the definition of $F$).
      
      For the inductive step from $s + 1$ to $s < N$, let $|u| = \ell s$. If there is some $w \in \Sigma^*$ of length $|w| = \ell \leq \ell (N - s) + F$ with $\bm{p} \cdot u \circ w \neq w$, we are done. Otherwise, we cannot have $\bm{p} \cdot uw = \idGrp$ in $\mathscr{G}(\mathcal{T})$ for all $w \in \Sigma^*$ of length $|w| = \ell$ (as this would imply $\bm{p} = \idGrp$ in $\mathscr{G}(\mathcal{T})$) and, thus, we may fix such a $w$ with $\bm{p} \cdot uw \neq \idGrp$ in $\mathscr{G}(\mathcal{T})$. Since we have $|uw| = \ell (s + 1)$, we may apply induction and obtain that there is some $w' \in \Sigma^*$ of length $|w'| \leq \ell (N - s - 1) + F$ with $\bm{p} \cdot uw \circ w' \neq w'$. This implies
      \[
        \bm{p} \cdot u \circ w w' = (\bm{p} \cdot u \circ w) (\bm{p} \cdot uw \circ w') \neq w w'
      \]
      for the word $w w'$ of length $|ww'| \leq \ell + \ell (N - s - 1) + F = \ell (N - s) + F$.
      
      At the end of the induction, we obtain from the case $s = 0$ that there is some witness $w \in \Sigma^*$ with $\bm{p} \circ w \neq w$ of length
      \begin{align*}
        |w| &\leq \ell N + F \leq \ell \left( \frac{\log \| \bm{p} \|}{\log \lambda} + 1 \right) + F = \frac{\ell}{\log \lambda} \log \| \bm{p} \| + \ell + F = A \log \| \bm{p} \| + B \\
        &\leq A \log | \bm{p} | + B \text{.}\qedhere
      \end{align*}
    \end{proof}
    
    \paragraph*{Balanced Iterated Commutators.}
    In addition to the normal commutator of two elements, we also need iterated commutators which we recursively split in the middle.
    \begin{definition}[{compare to \cite[Definition~3]{waechter2022automaton}}]\label{def:Bcommutator}
      For words $\bm{q}_0, \dots, \bm{q}_{2D - 1} \in Q^{\pm *}$ where $D = 2^d$ is a power of two, we define $B[\bm{q}_{D - 1}, \dots, \bm{q}_0]$ by induction on $d$ and let
      \begin{align*}
        B[\bm{q}_1] &= \bm{q}_1 \quad \text{and}\\
        B[\bm{q}_{2D - 1}, \dots, \bm{q}_0] &= \big[
          B[\bm{q}_{2D - 1}, \dots, \bm{q}_{D}], \:
          B[\bm{q}_{D - 1}, \dots, \bm{q}_0] \big] \textbf{.}
      \end{align*}
    \end{definition}
    \noindent{}This also immediately yields an operation $B[g_{D - 1}, \dots, g_0]$ for group elements $g_0, \dots, g_{D - 1}$ using the natural evaluation in the group.
  
    The reason for introducing balanced iterated commutators is that we may use them to simulate a $D$-ary logical conjunction in groups. The idea here is that the neutral element $\idGrp$ belongs to $\bot$ and all other elements are considered to belong to $\top$. One direction of the simulation then works in any group as we state in the following fact.\footnote{The fact can be proved using a simple induction on the structure of the balanced iterated commutators, see \cite[Fact~4]{waechter2022automaton}.}
    \begin{fact}[{compare to \cite[Fact~4]{waechter2022automaton}}]\label{fct:commutatorAndNegative}
      Let a group $G$ be generated by the alphabet $Q$ and let $\bm{q}_0, \dots, \bm{q}_{D - 1} \in Q^{\pm *}$ for some $D = 2^d$. If there is some $0 \leq i < D$ with $\bm{q}_i = \idGrp$ in $G$, we have $B[\bm{q}_{D - 1}, \dots, \bm{q}_0] = \idGrp$ in $G$.
    \end{fact}
  
    The reason that we use balanced iterated commutators (instead of the more common ones of the form $\big[ g_{D - 1}, [g_{D - 2}, \dots, g_0] \big]$) is that, this way, the depth remains logarithmic in the number of entries. This allows us to compute the balanced iterated commutator from its entries in logarithmic space.
    \begin{fact}[{compare to \cite[Lemma~7]{waechter2022automaton}}]\label{fct:BcommutatorIsLogSpaceComputable}
      The balanced commutator $B[ \bm{q}_{D - 1}, \dots, \bm{q}_0 ]$ can be computed from $\bm{q}_0, \dots, \bm{q}_{D - 1} \in Q^{\pm *}$ in logarithmic space.
    \end{fact}
  
    Normally, we cannot simply add balanced iterated commutators to cross diagrams and expect the resulting diagram to still hold. However, this is possible if all the entries act trivially on the input word (which can be seen by a simple induction on the structure of the balanced iterated commutators).
    \begin{fact}[{compare to \cite[Fact~8]{waechter2022automaton}}]\label{fct:commutatorInCrossDiagrams}
      Let $\mathcal{T} = (Q, \Sigma, \delta)$ be a \GAut, $u \in \Sigma^*$, $\bm{q}_0, \dots, \bm{q}_{D - 1} \in Q^{\pm *}$ with $D = 2^d$.  Then the cross diagram
      \begin{center}
        \begin{tikzpicture}[baseline=(m-4-1.base)]
          \matrix[matrix of math nodes, text height=1.25ex, text depth=0.25ex] (m) {
            & u & \\
            \bm{q}_0 &   & \bm{q}_0' \\
            & u & \\
            \vdots  & \vdots  & \vdots \\
            & u & \\
            \bm{q}_{D - 1} &   & \bm{q}_{D - 1}' \\
            & u & \\
          };
          
          \foreach \j in {1, 5} {
            \foreach \i in {1} {
              \draw[->] let
                \n1 = {int(2+\i)},
                \n2 = {int(1+\j)}
              in
                (m-\n2-\i) -> (m-\n2-\n1);
              \draw[->] let
                \n1 = {int(1+\i)},
                \n2 = {int(2+\j)}
              in
                (m-\j-\n1) -> (m-\n2-\n1);
            };
          };
        \end{tikzpicture}
        implies the diagram
        \begin{tikzpicture}[baseline=(m-2-3.base)]
          \matrix[matrix of math nodes, text height=1.25ex, text depth=0.25ex] (m) {
            & u & \\
            B[\bm{q}_{D - 1}, \dots, \bm{q}_0] &   & B[\bm{q}_{D - 1}', \dots, \bm{q}_0'] \\
            & u & \\
          };
          
          \foreach \j in {1} {
            \foreach \i in {1} {
              \draw[->] let
                \n1 = {int(2+\i)},
                \n2 = {int(1+\j)}
              in
               (m-\n2-\i) -> (m-\n2-\n1);
              \draw[->] let
                \n1 = {int(1+\i)},
                \n2 = {int(2+\j)}
              in
                (m-\j-\n1) -> (m-\n2-\n1);
            };
          };
        \end{tikzpicture}.
      \end{center}
    \end{fact}
    
    Another fact worth pointing out here is that (balanced iterated) commutators interact nicely with conjugation (which can be shown using a simple induction and the fact that $[h, g]^k = [h^k, g^k]$ for group elements $g, h, k$):
    \begin{fact}[{compare to \cite[Fact~6]{waechter2022automaton}}]\label{fct:commutatorConjugation}
      Let a group $G$ be generated by the alphabet $Q$ and let $\bm{q}_0, \dots, \bm{q}_{D - 1} \in Q^{\pm *}$ for some $D = 2^d$ as well as $\gamma \in Q^{\pm *}$.
      Then, we have:
      \[
        B[\bm{q}_{D - 1}, \dots, \bm{q}_0]^\gamma = B[\bm{q}_{D - 1}^\gamma, \dots, \bm{q}_0^{\gamma}] \text{ in } G
      \]
    \end{fact}
  
    \paragraph*{Simulating Logical Conjunctions.}
    We have already seen in \autoref{fct:commutatorAndNegative} that the balanced iterated commutator collapses to $\idGrp$ (which corresponds to $\bot$) if one of its entries is equal to $\idGrp$ (i.\,e.\ corresponds to $\bot$). This is one of the two directions to use the commutators as logical conjunctions. The other direction, however, does not hold for all elements of all groups. We will next look at some examples of groups where this approach does work.
    
    The first example we look at is the group $A_5$ of even permutations on the five-element set $\{ 1, \dots, 5 \}$. With regard to using $A_5$ for simulating logical conjunctions, we explicitly mention Barrington's work \cite{barrington89boundedWidth} (but also the predating work in \cite{makanin1984decidability, malcev1962equation, maurer1965property, krohn1966realizing}).
    \begin{figure}\centering
      \resizebox{\linewidth}{!}{
        \begin{tikzpicture}[level distance=1.5cm, auto,
            level 1/.style={sibling distance=8cm},
            level 2/.style={sibling distance=4cm},
            level 3/.style={sibling distance=2cm},
            level 4/.style={sibling distance=1cm}]
          \node[align=center] {{$\varepsilon$}\\\textcolor{gray}{$B[\bm{r}_7, \dots, \bm{r}_0]$}}
            child {node[align=center] {{1}\\\textcolor{gray}{$B[\bm{r}_7, \dots, \bm{r}_4]$}}
              child {node[align=center] {{11}\\\textcolor{gray}{$[\bm{r}_7, \bm{r}_6]$}}
                child {
                  node[align=center] {{111}\\\textcolor{gray}{$\bm{r}_7$}}
                }
                child {
                  node[align=center] {{011}\\\textcolor{gray}{$\bm{r}_6$}}
                }
              }
              child {node[align=center] {{01}\\\textcolor{gray}{$[\bm{r}_5, \bm{r}_4]$}}
                child {
                  node[align=center] {{101}\\\textcolor{gray}{$\bm{r}_5$}}
                }
                child {
                  node[align=center] {{001}\\\textcolor{gray}{$\bm{r}_4$}}
                }
              }
            }
           child {node[align=center] {{0}\\\textcolor{gray}{$B[\bm{r}_3, \dots, \bm{r}_0]$}}
             child {node[align=center] {{10}\\\textcolor{gray}{$[\bm{r}_3, \bm{r}_2]$}}
               child {
                 node[align=center] {{110}\\\textcolor{gray}{$\bm{r}_3$}}
                }
                child {
                  node[align=center] {{010}\\\textcolor{gray}{$\bm{r}_2$}}
                }
              }
              child {node[align=center] {{00}\\\textcolor{gray}{$[\bm{r}_1, \bm{r}_0]$}}
                child {
                  node[align=center] {{100}\\\textcolor{gray}{$\bm{r}_1$}}
                }
                child {
                  node[align=center] {{000}\\\textcolor{gray}{$\bm{r}_0$}}
                }
              }
            };
        \end{tikzpicture}
      }%
      \caption{Labeling the nodes of a regular binary tree with binary numbers in reverse notation with the corresponding commutators}\label{fig:binaryTree}
    \end{figure}
    \begin{example}[The Alternating Group of Degree $5$; {compare to \cite[Example~5]{waechter2022automaton}}]\label{ex:A5}
      In $A_5$ there is a non-identity element which is its own commutator (up to suitable conjugation). Such an element exists since there are two five-cycles in $A_5$ whose commutator is again a five-cycle and since five-cycles are always conjugate (see \cite[Lemma~1 and~3]{barrington89boundedWidth}).\footnote{More generally, $A_5$ is the smallest non-solvable group and its non-solvability is the main reason why this approach works.}
      
      More concretely, we may let $\sigma = (1 3 2 5 4)$, $\alpha = (2 3) (4 5)$ and $\beta = (2 4 5)$ (compare to \cite[Example~5]{waechter2022automaton}). A simple calculation shows that, with this choice, we have
      \[
        \sigma = [ \sigma^\beta, \sigma^\alpha ] \text{.}
      \]
      
      In order to handle the conjugation within the commutator (compare also to \cite[Remark~22]{waechter2022automaton}), we let $\revbin_d \, i$ denote the binary representation of the natural number $0 \leq i < 2^d$ of length (exactly) $d$ in reverse (i.\,e.\ with potentially leading zeros on the right; compare to \autoref{fig:binaryTree} as these numbers will eventually denote the leaves of a regular binary tree). Furthermore, let $f : \{ 0, 1 \}^* \to \{ \alpha, \beta \}^*$ be the homomorphism replacing each $0$ with $\alpha$ and each $1$ with $\beta$ and $\bm{k}_d(i) = f(\revbin_d \, i)$ for $0 \leq i < 2^d$. Note that we have
      \[
        \bm{k}_{d + 1}(i) = \begin{cases}
          \bm{k}_d(i) \alpha & \text{if } 0 \leq i < 2^d \\
          \bm{k}_d(i - 2^d) \beta & \text{if } 2^d \leq i < 2^{d + 1}
        \end{cases}
      \]
      since we have $\revbin_{d + 1} \, i = (\revbin_d \, i) 0$ for $0 \leq i < 2^d$ and $\revbin_{d + 1} \, i = (\revbin_d \, i - 2^d) 1$ for $2^d \leq i < 2^{d + 1}$.
      
      With this notation, we claim, for all $d$, that
      \[
        B[ \sigma^{\bm{k}_d(D - 1)}, \dots, \sigma^{\bm{k}_d(0)} ] = \sigma \text{ in } A_5
      \]
      where $D = 2^d$ and show it by induction. For $d = 0$ (i.\,e.\ $D = 2^0 = 1$), this is obvious as we have $B[ \sigma^\varepsilon ] = \sigma$ and, for the step from $d$ to $d + 1$ (i.\,e.\ from $D$ to $2D$), we have
      \begin{align*}
        &B\!\left[ \sigma^{\bm{k}_{d + 1}(2D - 1)}, \dots, \sigma^{\bm{k}_{d + 1}(0)} \right]\\
        =&
        \left[
          B\!\left[ \sigma^{\bm{k}_{d + 1}(2D - 1)}, \dots, \sigma^{\bm{k}_{d + 1}(D)} \right],
          B\!\left[ \sigma^{\bm{k}_{d + 1}(D - 1)}, \dots, \sigma^{\bm{k}_{d + 1}(0)} \right] 
        \right] \\
        =&
        \left[
          B\!\left[ \sigma^{\bm{k}_{d}(D - 1) \beta}, \dots, \sigma^{\bm{k}_{d}(0) \beta} \right],
          B\!\left[ \sigma^{\bm{k}_{d}(D - 1) \alpha}, \dots, \sigma^{\bm{k}_{d}(0) \alpha} \right] 
        \right] \\
        =&
        \left[
          B\!\left[ \sigma^{\bm{k}_{d}(D - 1)}, \dots, \sigma^{\bm{k}_{d}(0)} \right]^\beta,
          B\!\left[ \sigma^{\bm{k}_{d}(D - 1)}, \dots, \sigma^{\bm{k}_{d}(0)} \right]^\alpha
        \right] \\
        =&
        [ \sigma^\beta, \sigma^\alpha ] = \sigma
      \end{align*}
      in $A_5$ (where we have used \autoref{def:Bcommutator}, \autoref{fct:commutatorConjugation} and induction).
      
      If we fix some $d$ and $D = 2^d$, this allows us to define
      \[
        \bm{r}_{i} = \sigma^{\bm{k}_d(i)} = f(\revbin_d \, i)^{-1} \, \sigma \, f(\revbin_d \, i) \in \{ \sigma, \alpha, \beta \}^{\pm*}
      \]
      for $0 \leq i < D$ with $B[\bm{r}_{D - 1}, \dots, \bm{r}_{0}] = \sigma$ in $A_5$. It will become important later that these words can be computed in logarithmic space if $D$ is given in unary, i.\,e.\ that the function
      \function
        [the elements $\sigma, \alpha$ and $\beta$ of $A_5$]
        {the unary representation of a natural number $D$\newline
         with $D = 2^d$ for some $d$}
        {$\bm{r}_0, \dots, \bm{r}_{D - 1} \in \{ \sigma, \alpha, \beta \}^{\pm*}$ with $B[\bm{r}_{D - 1}, \dots, \bm{r}_{0}] = \sigma \neq \idGrp$ in $A_5$}\noindent
      is $\LogSpace$-computable.
      
      Using these elements, we may now simulate a $D$-ary logical conjunction in the group $A_5$ in the following sense. If we have a list of group elements $g_0, \dots, g_{D - 1} \in A_5$ with either $g_i = \bm{r}_i$ in $A_5$ (this is the case that $g_i$ corresponds to the truth value $\top$) or $g_i = \idGrp$ in $A_5$ (this is the case that $g_i$ corresponds to $\bot$) for all $0 \leq i < 2^d$, then we have in $A_5$
      \[
        B[g_{D - 1}, \dots, g_0] = \begin{cases}
          \sigma & \text{if } g_i = \bm{r}_i \text{ in } A_5 \text{ for all } 0 \leq i < D \\
          \idGrp & \text{otherwise}
        \end{cases}
      \]
      where the first case is the above claim and the second case follows by \autoref{fct:commutatorAndNegative}.
      In other words: if all the entries $g_i$ correspond to the truth value $\top$, the whole commutator also corresponds to $\top$ (in the sense that it is not $\idGrp$); on the other hand, if at least one entry $g_i$ corresponds to $\bot$ (i.\,e.\ is equal to $\idGrp$ in $A_5$), the whole commutator also corresponds to $\bot$ (i.\,e.\ is equal to $\idGrp$ in $A_5$).
      
      The group $A_5$ is finite (more precisely, it has 60 elements) and, thus, a finitary automaton group (as we discussed earlier). In fact, it is generated by the elements $\sigma$, $\alpha$ and $\beta$ (as a simple exhaustive calculation shows), which also shows that it is generated by the finitary \GAut
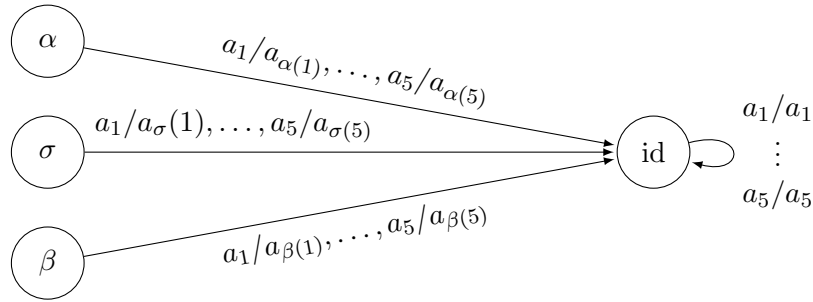
      \begin{center}
        \begin{tikzpicture}[auto, shorten >=1pt, >=latex, baseline=(b.base), node distance=0.5cm and 7cm]
          \node[state] (a) {$\alpha$};
          \node[state, below=of a] (s) {$\sigma$};
          \node[state, below=of s] (b) {$\beta$};
          \node[state, right=of s] (id) {$\id$};
          
          \path[->] (a) edge node[sloped] {$a_1 / a_{\alpha(1)}, \dots, a_5 / a_{\alpha(5)}$} (id)
                    (s) edge node[sloped, pos=0, above right] {$a_1 / {a_\sigma(1)}, \dots, a_5 / a_{\sigma(5)}$} (id)
                    (b) edge node[sloped, swap] {$a_1 / a_{\beta(1)}, \dots, a_5 / a_{\beta(5)}$} (id)
                    (id) edge[loop right] node[align=center] {$a_1 / a_1$\\\vdots\\$a_5 / a_5$} (id);
        \end{tikzpicture}
      \end{center}
      with alphabet size five (instead of using the general construction stated above to generate finite groups by finitary automata, which would yield an alphabet size of $|A_5| = 60$). Here, it is important to point out that the action of the state $\sigma$ ($\alpha$, $\beta)$ is indeed to apply the permutation $\sigma$ ($\alpha$, $\beta$) on the first letter and then to act as the identity, which justifies the double usage of the name here and allows us to observe that the above results about the commutator(s) also hold if we consider the $\bm{r}_i$ as state sequences over this automaton.
    \end{example}
  
    The group $A_5$ is not the only automaton group where the commutator approach for simulating logical conjunctions works (and where this is also \enquote{efficiently} usable algorithmically). In fact, Bartholdi, Figelius, Lohrey and Weiß introduced the class of {(uniformly) strongly efficiently non-solvable} (SENS) groups for this, which we will look at next.
    \begin{example}[SENS Groups]\label{ex:uSENS}
      A group $G$ finitely generated by $R$ is \emph{uniformly strongly efficiently non-solvable} (uniformly SENS)\footnote{Please note that we have slightly changed the indices in our definition compared to \cite{bartholdi2019groups}.} if there are a constant $\mu \in \mathbb{N}$ and words $\bm{r}_{d,v} \in R^{\pm*}$ for all $d \in \mathbb{N}$, $v \in \{ 0,1 \}^{\leq d}$ such that
      \begin{enumerate}[label=(\alph*), ref=(\alph*)]
        \item\label{SENSa} $|\bm{r}_{d,v}| \leq 2^{\mu d}$ for all $v \in \{ 0,1 \}^{d}$,
        \item\label{SENSb} $\bm{r}_{d,v} = \bigl[ \bm{r}_{d, 1v},\, \bm{r}_{d, 0v} \bigr]$ for all $v \in \{ 0,1 \}^{< d}$ (here we take the commutator of words)\footnote{Compare this to the tree in \autoref{fig:binaryTree}.},
        \item\label{SENSc} $\bm{r}_{d, \varepsilon} \neq \idGrp$ in $G$ and
        \item\label{SENSu} given $v \in \{0,1\}^d$, a positive integer $i$ encoded in binary with $\mu d$ bits, and $a \in R^{\pm 1}$ one can decide in $\ComplexityClass{DLinTime}$ whether the $i^\textnormal{th}$ letter of $\bm{r}_{d,v}$ is $a$. Here, \ComplexityClass{DLinTime} is the class of problems decidable in linear time on a Turing machine with random access to its input.
      \end{enumerate}
      The conditions \ref{SENSa} and \ref{SENSu} together imply that function
      \function
        {the unary representation of a natural number $D$\newline
          with $D = 2^d$ for some $d$}
        {$\bm{r}_{0}, \dots, \bm{r}_{D - 1} \in R^{\pm*}$ with 
         $B[\bm{r}_{D - 1}, \dots, \bm{r}_{0}] \neq \idGrp$ in $G$}\noindent
      is \LogSpace-computable by letting $\bm{r}_i = \bm{r}_{d, \revbin_d\,i}$ (with the notation $\revbin_d\,i$ for the binary representation of length exactly $d$ of $i$ in reverse; compare to \autoref{ex:A5} and \autoref{fig:binaryTree}). The reason for this is that $d$ is logarithmic in the input size $D$ and we may thus store any value $v \in \{ 0, 1 \}^d$. Since the length of $\bm{r}_{d, v}$ for $v \in \{ 0, 1 \}^d$ does not exceed $2^{\mu d}$ (by \ref{SENSa}), we may also store a counter in binary counting up to this length. Using this counter and the $\ComplexityClass{DLinTime}$-algorithm, we may now test which element of $R^{\pm 1}$ should be outputted at each position of $\bm{r}_{v, d}$. This algorithm runs in time $\mu d$ and, thus, in particular cannot require more space than $\mu d$.
      
      The condition \ref{SENSb} simply reflects the inductive definition of our balanced iterated commutators (from \autoref{def:Bcommutator}) and condition \ref{SENSc} then ensures
      \[
        B[ \bm{r}_{D - 1}, \dots, \bm{r}_{0} ] \neq \idGrp \text{ in } G \text{.}
      \]
      
      Now, we can use the same idea as in the case of $A_5$. Suppose we have a list $g_0, \dots, g_{D - 1}$ of group elements from $G$ where $D = 2^d$ such that, for all $0 \leq i < D$, we either have $g_i = \bm{r}_{i}$ in $G$ (corresponding to $\top$) or $g_i = \idGrp$ in $G$ (corresponding to $\bot$). Then, we have $B[ g_{D - 1}, \dots, g_0 ] \neq \idGrp$ in $G$ if we have $g_i = \bm{r_i}$ in $G$ for all $i$ and we have $B[ g_{D - 1}, \dots, g_0 ] = \idGrp$ in $G$ otherwise (i.\,e.\ if $g_i = \idGrp$ in $G$ for some $i$) by \autoref{fct:commutatorAndNegative}.
    \end{example}
    
    With regard to automaton groups, the class of uniformly SENS groups includes not only $A_5$ but also (see \cite{bartholdi2019groups}) the free group of rank three (generated by the Aleshin automaton; see \cite[Example~21]{waechter2022automaton} and the references therein) and Grigorchuk's group (from \autoref{ex:grigorchuk}). These two groups are of particular interest because their generating automata use a binary alphabet.
    This is interesting because any automaton generating $A_5$ must use an alphabet with at least five elements as, in fact, every group generated by an automaton with alphabet size at most four is solvable (and we, thus, cannot directly use the above commutator idea; see, for example, \cite{waechter2024word} for this).
    Since the automaton generating Grigorchuk's group is additionally contracting (we discussed this in \autoref{ex:grigorchuk}), this group will play an important role in our construction later on
    and we will describe how to directly construct the entries of a non-identity balanced iterated commutator in Grigorchuk's group next (based on \cite[Proposition~5.17]{bartholdi2019groups}).
    
    \begin{figure}\centering
      \begin{subfigure}{\dimexpr\linewidth-5cm}\centering
        \begin{tikzpicture}[auto, shorten >=1pt, >=latex, node distance=1cm and 1.5cm]
          \node[state, ellipse] (d+1) {$B_x(d + 1)$};
          \node[state, right=of d+1] (n) {};
          \node[state, right=of n, ellipse, dashed] (d) {$B_x(d)$};
          \node[state, below=of n] (id) {$\id$};
          
          \path[->] (d+1) edge node {$1/1$} (n)
                          edge node[sloped] {$0/0$} (id)
                    (n) edge node {$1/1$} (d)
                        edge node {$0/0$} (id)
                    (id) edge[loop right] node[align=center] {$0/0$\\$1/1$} (id)
                    ;
        \end{tikzpicture}
        \caption{Inductive description of $B_x(d + 1)$ as an automaton}\label{sfig:BxAutomaton}
      \end{subfigure}%
      \begin{subfigure}{5cm}\centering
        \begin{tikzpicture}[level distance=1cm, auto,
          level 1/.style={sibling distance=1cm},
          level 2/.style={sibling distance=1cm}]
          \node[align=center] {$\cdot$}
            child {node {$\idGrp$}}
            child {node {$\cdot$}
              child {node {$\idGrp$}}
              child {node {$B_x(d)$}}
            }
            ;
        \end{tikzpicture}
        \caption{Inductive description of $B_x(d + 1)$ as a tree automorphism}\label{sfig:BxTree}
      \end{subfigure}
    \end{figure}
    \begin{example}[Again: Grigorchuk's Group]\label{ex:grigorchukCommutator}
      Recall Grigorchuk's group and its generating automaton from \autoref{ex:grigorchuk}. Since every state is its own inverse in Grigorchuk's group, the inverse of a state sequence $\bm{p} = p_\ell \dots p_1$ (where $p_1, \dots, p_\ell$ are states) is its reverse as a word, i.\,e.\ $p_1 \dots p_\ell$.
      
      In order to find arbitrarily deep balanced iterated commutators in Grigorchuk's group, we make an inductive definition and let
      \begin{align*}
        B_x(0) &= (abad)^2 = x &
        B_y(0) &= b B_x(0) b = y \\
        B_{\bar{x}}(0) &= (daba)^2 = \bar{x} &
        B_{\bar{y}}(0) &= b B_{\bar{x}}(0) b =\bar{y}\\
      \shortintertext{as well as}
        B_x(d + 1) &= \left[ B_{\bar{x}}(d), B_{\bar{y}}(d) \right] &
        B_y(d + 1) &= \left[ B_y(d), B_x(d) \right] \\
        B_{\bar{x}}(d + 1) &= \left[ B_{\bar{y}}(d), B_{\bar{x}}(d) \right] &
        B_{\bar{y}}(d + 1) &= \left[ B_{x}(d), B_{y}(d) \right] &
      \end{align*}
      for all $d \geq 0$. This inductively describes a (\LogSpace-computable) procedure to eventually obtain the state sequences which may be used as the entries for a balanced iterated commutator that is not equal to $\idGrp$ in the group (e.\,g.\ using $B_x(d)$). Note that these state sequences correspond to the leaves in \autoref{fig:binaryTree} and are either $x, y, \bar{x}$ or $\bar{y}$. This commutator can then be used as a logical conjunction in the same way as in \autoref{ex:uSENS} (and \autoref{ex:A5}).
      
      That the commutators $B_x(d)$, $B_y(d)$, $B_{\bar{x}}(d)$ and $B_{\bar{y}}(d)$ are all different to $\idGrp$ in the group may be seen using an induction on $d$ (and some calculations in the group), which we will not explicitly demonstrate here. Instead, we refer the reader to \cite[Proposition~5.17]{bartholdi2019groups} for details and only mention that the idea for the induction is that $B_x(d)$ is \enquote{almost} equal to $x$ in the group (with analogous statements for $y, \bar{x}$ and $\bar{y}$, respectively). By this, we mean that $B_x(d + 1)$ may (loosely) be described by the automaton in \autoref{sfig:BxAutomaton}
      where the state $B_x(d)$ is given inductively
      or, for readers familiar with the interpretation of elements of automaton groups as tree automorphisms, by the automorphism of the infinite binary regular tree depicted in \autoref{sfig:BxTree} where the subtrees on the first two levels are not permuted.
    \end{example}
  \end{section}

  \begin{section}{The Word Problem}
    The uniform word problem for finitary automaton groups is the decision problem
    \problem{
      a finitary \GAut $\mathcal{T} = (Q, \Sigma, \delta)$ and\newline
      a state sequence $\bm{q} \in Q^{\pm *}$
    }{
      is $\bm{q} = \idGrp$ in $\mathscr{G}(\mathcal{T})$?
    }\noindent
    In this section, we will show that it is \coNP-complete and even that it remains so if we limit the alphabet size to $|\Sigma| = 2$.
    We start with the easier part and first show that the problem is in \coNP.
    
    \begin{proposition}\label{prop:finitaryWPisInCoNP}
      The uniform word problem for finitary automaton groups is in \coNP.
    \end{proposition}
    \begin{proof}
      We solve the complement of the problem by a guess and check approach in \NP. First, we guess a witness $u$ on which $\bm{q}$ acts non-trivially. The length of a shortest such witness is at most the depth of the automaton $\mathcal{T}$, which, in turn, is bounded by the size of $\mathcal{T}$. Thus, the witness can be guessed in linear time.
      
      Then, we compute $u_i = q_i \dots q_1 \circ u$ for $\bm{q} = q_\ell \dots q_1$ (with $q_i \in Q^{\pm 1}$, $1 \leq i \leq \ell$) state by state. This requires time $|\bm{q}| \cdot |u|$ and is, thus, certainly possible in polynomial time.
    \end{proof}
    
    \paragraph*{$\mathcal{R}$-Finitary Automata.}
    For the other direction, we show a stronger result, namely that the problem remains \coNP-hard if we restrict ourselves to a binary input/output alphabet for the finitary automata. Here, it is convenient to first define a notion that is very similar to being finitary: instead of considering automata where we always reach an/the identity state after reading sufficiently many letters, we may consider the more general notion where we reach a specific subautomaton instead.
    \begin{definition}
      Let $\mathcal{T} = (Q, \Sigma, \delta)$ and $\mathcal{R} = (R, \Gamma, \varrho)$ be \GAuta. We say that $\mathcal{T}$ is \emph{$\mathcal{R}$-finitary} if
      \begin{enumerate}
        \item $\mathcal{T}$ and $\mathcal{R}$ have the same alphabet (i.\,e.\ $\Sigma = \Gamma$),
        \item $\mathcal{R}$ is a subautomaton of $\mathcal{T}$ (more precisely: $R \subseteq Q$ and $\varrho = \delta \cap R \times \Gamma \times \Gamma \times R$) and
        \item there is a constant $d$ such that, for all $u \in \Sigma^*$ with $|u| \geq d$ and all state sequences $\bm{q} \in Q^{\pm *}$, we have $\bm{q} \cdot u \in R^{\pm *}$.
      \end{enumerate}
      The minimal such $d$ is called the \emph{$\mathcal{R}$-depth} of $\mathcal{T}$.
    \end{definition}
    \begin{remark}
      Note that, in an $\mathcal{R}$-finitary \GAut $\mathcal{T}$, it is impossible to leave the subautomaton $\mathcal{R}$ by reading any word since its alphabet is the same as that of $\mathcal{T}$ (and $\mathcal{R}$ must be deterministic and complete). This show that the $\mathcal{R}$-depth of $\mathcal{T}$ is always bounded by its size (actually we may even subtract the size of $\mathcal{R}$).
      
      We may visualize an $\mathcal{R}$-finitary \GAut $\mathcal{T}$ as an automaton with two parts: there is the subautomaton $\mathcal{R}$ and then there may also be a part of transitions that are directed towards this subautomaton.
    \end{remark}

    \begin{remark}
      Clearly, a \GAut $\mathcal{T} = (Q, \Sigma, \delta)$ is finitary if and only if it is $\mathcal{E}$-finitary for the \GAut $\mathcal{E} = (\{ \id \}, \Sigma, \{ \trans{\id}{a}{a}{\id} \mid a \in \Sigma \})$ and the $\mathcal{E}$-depth is exactly the depth of $\mathcal{T}$ as a finitary automaton.
      
      Additionally, if $\mathcal{T}$ is $\mathcal{S}$-finitary and $\mathcal{S}$ itself is $\mathcal{R}$-finitary, this implies that $\mathcal{T}$ is also $\mathcal{R}$-finitary. In particular, we have that an $\mathcal{R}$-finitary \GAut is finitary if $\mathcal{R}$ is.
    \end{remark}
    
    Being finitary (or, more generally, the activity; see e.\,g.\ \cite{waechter2024word}) is not the only property an $\mathcal{R}$-finitary \GAut inherits from $\mathcal{R}$. Most interesting to us will be later that, if $\mathcal{R}$ is contracting, we still have witnesses of logarithmic length to prove that a state sequence acts non-trivially (compare  \autoref{fct:contractingLogarithmicWitness}).
    \begin{fact}\label{lem:finitarilyContracting}
      Let $\mathcal{R}$ be a contracting \GAut. Then, there are constants $A$ and $B$ such that, for every $\mathcal{R}$-finitary \GAut $\mathcal{T} = (Q, \Sigma, \delta)$ and every state sequence $\bm{q} \in Q^{\pm *}$ with $\bm{q} \neq \idGrp$ in $\mathscr{G}(\mathcal{T})$, there is a witness $w \in \Sigma^*$ of length $|w| \leq |Q| + A \log | \bm{q} | + B$ with $\bm{q} \circ w \neq w$.
    \end{fact}
    \begin{proof}
      Let $\mathcal{R} = (R, \Sigma, \varrho)$ and let $A$ and $B$ be the constants from \autoref{fct:contractingLogarithmicWitness} with respect to $\mathcal{R}$.
      Consider some state sequence $\bm{q} \in Q^{\pm *}$ with $\bm{q} \neq_{\mathcal{T}} \varepsilon$. We are done if there is some $w \in \Sigma^*$ with $|w| \leq |Q|$ and $\bm{q} \circ w \neq w$. Therefore, assume that we have $\bm{q} \circ w = w$ for all $w \in \Sigma^*$ of length $|w| = |Q|$. We cannot have $\bm{q} \cdot w =_{\mathcal{T}} \varepsilon$ for all of them (since this would imply $\bm{q} =_{\mathcal{T}} \varepsilon$), so we may fix some $w_0 \in \Sigma^*$ with $\bm{q} \cdot w_0 \neq_{\mathcal{T}} \varepsilon$ and $|w_0| = |Q|$. Since the $\mathcal{R}$-depth of $\mathcal{T}$ is bounded by $|Q|$, we have $\bm{q} \cdot w_0 = \bm{r} \in R^{\pm *}$, which yields the (black part of the) cross diagram
      \begin{center}
        \begin{tikzpicture}[baseline=(m-2-1.base)]
          \matrix[matrix of math nodes, text height=1.25ex, text depth=0.25ex] (m) {
              & w_0 && \textcolor{gray}{w_1} & \\
            \bm{q} && \bm{r} && {} \\
              & w_0 && \textcolor{gray}{\neq w_1} &\\
          };
          \draw[->] (m-2-1) -> (m-2-3);
          \draw[->] (m-1-2) -> (m-3-2);
          
          \draw[->, gray] (m-2-3) -> (m-2-5);
          \draw[->, gray] (m-1-4) -> (m-3-4);
        \end{tikzpicture}
      \end{center}
      and (by \autoref{fct:contractingLogarithmicWitness}) a witness $w_1 \in \Sigma^*$ of length $|w_1| \leq A \log | \bm{r} | + B = A \log |\bm{q}| + B$ with $\bm{r} \circ w_1 \neq w_1$, which yields the gray additions to the above diagram (where $\neq w_1$ is some word different to $w_1$ and the state sequence on the right is not of interest) and also $w_0 w_1$ as the witness of length $|w_0 w_1| \leq |Q| + A \log |\bm{q}| + B$ for $\bm{q} \circ w_0 w_1 \neq w_0 w_1$.
    \end{proof}
    
    \paragraph*{The Main Reduction.}
    With the definition of an $\mathcal{R}$-finitary automaton at hand, we can now proceed to the main part of our proof where we reduce the satisfiability problem for boolean formulae to the uniform word problem for families of certain $\mathcal{R}$-finitary \GAuta.
    \begin{theorem}\label{thm:coNPHardFamily}
      Let $\mathcal{R}$ be a \GAut with a state $\id$ acting like the identity for which the function
      \function
        [the \GAut $\mathcal{R} = (R, \Sigma, \varrho)$]
        {the unary representation of a natural number $D$\newline
          with $D = 2^d$ for some $d$}
        {$\bm{r}_1, \dots, \bm{r}_D$ with $B[\bm{r}_D, \dots, \bm{r}_1] \neq \idGrp$ in $\mathscr{G}(\mathcal{R})$}\noindent
      is $\LogSpace$-computable.
      Then, the uniform word problem for the family of $\mathcal{R}$-finitary \GAuta
      \problem
        {an $\mathcal{R}$-finitary \GAut $\mathcal{T} = (Q, \Sigma, \delta)$ and\newline
         a state sequence $\bm{q} \in Q^{\pm*}$}
        {is $\bm{q} = \idGrp$ in $\mathscr{G}(\mathcal{T})$?}\noindent
      is \coNP-hard (under many-one $\LogSpace$-reductions).
    \end{theorem}

    \begin{proof}
      First, observe that any \GAut over a single element alphabet can only generate the trivial group (which yields $B[\bm{r}_D, \dots, \bm{r}_1] = \idGrp$ in $\mathscr{G}(\mathcal{R})$ for all $D$). Therefore, we may assume $|\Sigma| \geq 2$.
      
      \pagebreak
      We reduce the \NP-hard\footnote{This is a well-known problem from Karp's list of \NP-complete problem, see e.\,g.\ \cite[Problem~9.5.5]{papadimitriou97computational}.} satisfiability problem for boolean formulae
      \problem{
        a boolean formula $\varphi$ in $3$-conjunctive normal form
      }{
        is $\varphi$ satisfiable?
      }\noindent
      to the complement of the stated problem by using a many-one \LogSpace-reduction. In other words, we need to map (in logarithmic space) a boolean formula $\varphi$ in $3$-conjunctive normal form over a set of variables $\mathbb{X} = \{ x_1, \dots, x_N \}$ to an $\mathcal{R}$-finitary \GAut $\mathcal{T}$ and a state sequence $\bm{q}$ such that $\bm{q}$ does \textbf{not} act as the identity if and only if $\varphi$ is satisfiable.

      As $\varphi$ is in $3$-conjunctive normal form, we may write $\varphi = \bigwedge_{k = 1}^K C_k$ where every clause $C_k$ contains exactly three distinct literals over $\mathbb{X}$. Without loss of generality, we may assume that no clause contains the same variable as a positive and a negative literal (as such clauses are satisfied by all assignments and can, thus, be dropped). In other words, we have $C_k = (\lnot) x_{n_1} \lor (\lnot) x_{n_2} \lor (\lnot) x_{n_3}$ for three pairwise distinct $n_1$, $n_2$ and $n_3$ with $1 \leq n_1 < n_2 < n_3 \leq N$. Additionally, we may assume that the number of clauses $K$ is a power of two. We may do this since we can easily just repeat one of the clauses and only have to count up to the next power of two (which can be done in logarithmic space).
      
      The alphabet of $\mathcal{T}$ must be $\Sigma$ (since it needs to be $\mathcal{R}$-finitary) and we identify two arbitrary letters with $\bot$ and $\top$, respectively. This allows us to encode an assignment $\mathcal{A}: \mathbb{X} \to \mathbb{B}$ as the word $\langle \mathcal{A} \rangle = \mathcal{A}(x_N) \dots \mathcal{A}(x_1)$ of length $N$.\footnote{Note that the right-most letter here corresponds to the first variable $x_1$. We could have done this the other way round as well but it turns out that this numbering has some technical advantages.} Note that a word $w \in \Sigma^*$ of length $N$ encodes an assignment (i.\,e.\ $w = \langle \mathcal{A} \rangle$ for some assignment $\mathcal{A}$) if and only if $w \in \{ \bot, \top \}^*$.
      
      The general idea is now that we check for every clause $C_k$ whether the first $N$ letters of the input form an encoding of an assignment satisfying $C_k$. If this is not the case (i.\,e.\ if a letter different to $\bot$ and $\top$ appears or if the encoded assignment does not satisfy $C_k$), we will go into an identity state, which can be thought of as a \enquote{fail} state. Otherwise, we will end up in the state sequence belonging to the $k$-th entry in the balanced iterated commutator (from \autoref{def:Bcommutator}) $B[\bm{r}_K, \dots, \bm{r}_1] \neq \idGrp$ in $\mathscr{G}(\mathcal{R})$. This allows us to finally use the balanced commutator to make a conjunction of all these checks.
      
      By hypothesis, we may compute the entries $\bm{r}_1, \dots, \bm{r}_K$ for the balanced iterated commutator $B[\bm{r}_K, \dots, \bm{r}_1] \neq \idGrp$ in $\mathscr{G}(\mathcal{R})$ in \LogSpace. Since \LogSpace-computable functions are closed under composition, we may assume that the $\bm{r}_k$ are already part of the input (alongside the formula $\varphi$) for our reduction.
      
      We will give a precise definition of the automaton $\mathcal{T} = (Q, \Sigma, \delta)$ by describing various parts.
      
      First, we need $\mathcal{R} = (R, \Sigma, \varrho)$ as a subautomaton of $\mathcal{T}$ (i.\,e.\ $R \subseteq Q$, $\varrho \subseteq \delta$), so we output this part first. This already yields the identity state $\id \in Q$ (with the transitions $\{ \trans{\id}{a}{a}{\id} \mid a \in \Sigma \} \subseteq \delta$). In order to simplify out notation later on, we identify every state $r \in R$ with $r_0 = r \in R \subseteq Q$.
      
      \begin{figure}\centering
        \begin{tikzpicture}[auto, shorten >=1pt, >=latex]
          
          \node[state] (N) {$r_N$};
          \node[state, ellipse, inner sep=0pt, right=of N] (N-1) {$r_{N - 1}$};
          \node[right=of N-1] (dots) {$\dots$};
          \node[state, right=of dots, dotted] (0) {$r_{0}$};
          \node[state, below=1cm of dots, dotted] (id) {$\id$};
          
          \draw[->] (N) edge node[above] {$\bot / \bot$} node[below] {$\top / \top$} (N-1)
                    (N) edge[bend right] node[sloped] {$b / b$} (id)
                    (N-1) edge node[above] {$\bot / \bot$} node[below] {$\top / \top$} (dots)
                    (N-1) edge[bend right, out=-45] node[sloped] {$b / b$} (id)
                    (dots) edge node[above] {$\bot / \bot$} node[below] {$\top / \top$} (0)
                    (dots) edge node[sloped, swap] {$b / b$} (id)
                    ;
          \draw[->, dotted] (id) edge[loop right] node {$a / a$} (id);
        \end{tikzpicture}
        \caption{The automaton part for the states $\{ r_n \mid 0 < n \leq N \}$. The dotted states are already defined by the subautomaton $\mathcal{R}$ and the transitions exist for all $a \in \Sigma$ and $b \in \Sigma \setminus \{ \bot, \top \}$.}\label{fig:sigmaPart}
      \end{figure}
      Then, for every $r \in R$, we define a state $r_N$ that check whether the first $N$ letters are either $\bot$ or $\top$ and, if this is the case, goes to $r = r_0$ afterwards (i.\,e.\ it acts like $r$ starting from the $(N + 1)$-th letter). Otherwise, it will go to the identity state as a \enquote{fail} state. For this, we use the states $\{ r_n \mid 0 < n \leq N, r \in R \} \subseteq Q$ together with the transitions
      \begin{align*}
        &\left\{
          \trans{r_n}{\bot}{\bot}{r_{n - 1}},
          \trans{r_n}{\top}{\top}{r_{n - 1}},
          \trans{r_n}{b}{b}{\id} \mid r \in R, 0 < n \leq N, b \in \Sigma \setminus \{ \bot, \top \} \right\}
        \subseteq \delta \text{.}
      \end{align*}
      See \autoref{fig:sigmaPart} for a graphical representation. By construction, we obtain for all $0 \leq n \leq N$ the cross diagram 
      \begin{equation}
        \begin{tikzpicture}[baseline=(m-2-1.base)]
          \matrix[matrix of math nodes, text height=1.25ex, text depth=0.25ex, ampersand replacement=\&] (m) {
                     \& w \& \\
            r_n \&   \& |[align=left]| $\begin{cases}
              r_0 & \text{if } w \in \{ \bot, \top \}^* \\
              \id & \text{otherwise}
            \end{cases}$ \\
                     \& w \& \\
          };
          \foreach \j in {1} {
            \foreach \i in {1} {
              \draw[->] let
                \n1 = {int(2+\i)},
                \n2 = {int(1+\j)}
              in
                (m-\n2-\i) -> (m-\n2-\n1);
              \draw[->] let
                \n1 = {int(1+\i)},
                \n2 = {int(2+\j)}
              in
                (m-\j-\n1) -> (m-\n2-\n1);
            };
          };
        \end{tikzpicture}\label{eqn:sigmaCrossDiagram}
      \end{equation}
      for all $w \in \Sigma^*$ of length $n$. Recall that, for a word $w \in \Sigma^*$ of length $N$, we have $w = \langle \mathcal{A} \rangle$ for some assignment $\mathcal{A}$ if and only if $w \in \{ \bot, \top \}^*$ (i.\,e.\ if we are in the upper case in the above diagram).
      We have, in particular, that $r_N$ does not change the first $N$ letters.
      Note that this part is $\mathcal{R}$-finitary (with $\mathcal{R}$-depth $N$) and may be computed in logarithmic space (as we only need to count up to the value $N$ in binary).
      
      Most interesting are those parts of the automaton which are used to verify whether a clause is satisfied. In order to describe these parts, consider the clause $C_k = L_1 \lor L_2 \lor L_3$ for all $1 \leq k \leq K$ where $L_i$ for $i \in \{ 1, 2, 3 \}$ is either a positive or a negative literal of a variable $x_{n_i}$. Without loss of generality, we may assume $1 \leq n_1 < n_2 < n_3 \leq N$ and we say that $x_{n}$ appears \emph{positively} in $C_k$ if $L_i = x_{n}$ and it appears \emph{negatively} in $C_k$ if $L_i = \lnot x_{n}$ (for some $i \in \{ 1, 2, 3 \}$). If a variable appears neither positively nor negatively, we say that it does \emph{not} appear in $C_k$.
      
      \begin{figure}\centering
        \resizebox{\linewidth}{!}{
          \begin{tikzpicture}[auto, shorten >=1pt, >=latex, baseline=(ck.base), node distance=1.25cm and 0.75cm]
            \node[state, ellipse, inner sep=0pt] (ck) {$c_{k, r, N}$};
            \node[right=of ck] (dots) {$\dots$};
            \node[state, ellipse, inner sep=0pt, right=of dots] (u0) {$c_{k, r, n_3}$};
            \node[state, ellipse, inner sep=0pt, right=1.25cm of u0] (u1) {$c_{k, r, n_3 - 1}$};
            \node[state, ellipse, inner sep=0pt, below=of u1, dashed] (l1) {$r_{n_3 - 1}$};
            \node[right=of u1] (udots1) {$\dots$};
            \node[anchor=base] at ($(l1.base-|udots1.base)$) (ldots1) {$\dots$};
            \node[state, ellipse, inner sep=0pt, right=of udots1] (u2) {$c_{k, r, n_2}$};
            \node[state, dashed, anchor=base] at ($(l1.base-|u2.base)$) (l2) {$r_{n_2}$};
            \node[state, ellipse, inner sep=0pt, right=1.25cm of u2] (u3) {$c_{k, r, n_2 - 1}$};
            \node[state, ellipse, inner sep=0pt, dashed, anchor=base] at ($(l1.base-|u3.base)$) (l3) {$r_{n_2 - 1}$};
            \node[right=of u3] (udots3) {$\dots$};
            \node[anchor=base] at ($(l1.base-|udots3.base)$) (ldots3) {$\dots$};
            \node[state, right=of udots3, dotted] (id) {$\id$};
            \node[state, right=of ldots3, dotted] (sigma) {$r_0$};
            
            \draw[->] (ck) edge node[above] {$\bot/\bot$} node[below] {$\top/\top$} (dots)
                      (dots) edge node[above] {$\bot/\bot$} node[below] {$\top/\top$} (u0)
                      (u0) edge[thick] node[above] {$\bm{\bot/\bot}$} (u1)
                           edge[thick] node[sloped, anchor=center, below] {$\bm{\top/\top}$} (l1)
                      (u1) edge node[above] {$\bot/\bot$} node[below] {$\top/\top$} (udots1)
                      (udots1) edge node[above] {$\bot/\bot$} node[below] {$\top/\top$} (u2)
                      (u2) edge[thick] node[above] {$\bm{\top/\top}$} (u3)
                           edge[thick] node[sloped, anchor=center, above, pos=0.6, yshift=-0.25ex, xshift=-0.25ex] {$\bm{\bot/\bot$}} (l3)
                      (u3) edge node[above] {$\bot/\bot$} node[below] {$\top/\top$} (udots3)
                      (udots3) edge node[above] {$\bot/\bot$} node[below] {$\top/\top$} (id)
            ;
            \draw[->, dashed]
                      (l1) edge node[above] {$\bot/\bot$} node[below] {$\top/\top$} (ldots1)
                      (ldots1) edge node[above] {$\bot/\bot$} node[below] {$\top/\top$} (l2)
                      (l2) edge node[above] {$\bot/\bot$} node[below] {$\top/\top$} (l3)
                      (l3) edge node[above] {$\bot/\bot$} node[below] {$\top/\top$} (ldots3)
                      (ldots3) edge node[above] {$\bot/\bot$} node[below] {$\top/\top$} (sigma)
            ;
            
            \node[anchor=south] at ($(u0.north)!0.5!(u1.north)$) (Xi1) {$x_{n_3}$};  
            \path[fill=gray, opacity=0.2, rounded corners] ($(u0.east |- Xi1.north)-(1ex,-0.5ex)$) rectangle ($(l1.south -| l1.west)+(1ex,-0.5ex)$);
            
            \node[anchor=south] at ($(u2.north)!0.5!(u3.north)$) (Xi2) {$x_{n_2}$};  
            \path[fill=gray, opacity=0.2, rounded corners] ($(u2.east |- Xi2.north)-(1ex,-0.5ex)$) rectangle ($(l3.south -| l3.west)+(1ex,-0.5ex)$);
          \end{tikzpicture}
        }
        \caption{Part of the automaton for the states $\{ c_{k, r, n} \mid 0 < n \leq N \}$ (with $r \in R$, $1 \leq k \leq K$). We assume $x_{n_3}$ to appear positively in $C_k$ while $x_{n_2}$ is assumed to appear negatively. The part for $x_{n_1}$ is not drawn for space reasons. Dashed states and transitions are already defined above, dotted states are part of the subautomaton $\mathcal{R}$. Missing transitions are of the form $b/b$ and go to $\id$ (for $b \in \Sigma \setminus \{ \bot, \top \}$).}\label{fig:cknPart}
      \end{figure}
      Now, in order to verify that the clause $C_k$ is satisfied, we use the states $\{ c_{k, r, n} \mid r \in R, 0 < n \leq N \} \subseteq Q$ with the transitions
      \begin{align*}
        &\{ \trans{c_{k, r, n}}{\bot}{\bot}{c_{k, r, n - 1}}, \trans{c_{k, r, n}}{\top}{\top}{c_{k, r, n - 1}}
          \mid
            \begin{aligned}[t]
              &r \in R, 0 < n \leq N,\\[-0.75ex]
              &x_n \text{ does not appear in } C_k \}
            \end{aligned}
           \\
        {}\cup{}& \{ \trans{c_{k, r, n}}{\bot}{\bot}{c_{k, r, n - 1}}, \trans{c_{k, r, n}}{\top}{\top}{r_{n - 1}}
          \mid
          \begin{aligned}[t]
            &r \in R, 0 < n \leq N,\\[-0.75ex]
            &x_n \text{ appears positively in } C_k \}
          \end{aligned}
          \\
        {}\cup{}& \{ \trans{c_{k, r, n}}{\bot}{\bot}{r_{n - 1}}, \trans{c_{k, r, n}}{\top}{\top}{c_{k, n - 1}}
          \mid
          \begin{aligned}[t]
            &r \in R, 0 < n \leq N,\\[-0.75ex]
            &x_n \text{ appears negatively in } C_k \}
          \end{aligned}
          \\
        {}\cup{}& \{ \trans{c_{k, r, n}}{b}{b}{\id}
          \mid r \in R, 0 < n \leq N, b \in \Sigma \setminus \{ \bot, \top \} \}
        \subseteq \delta
      \end{align*}
      where we identify $c_{k, r, 0}$ with the identity state $\id$.
      This results in the automaton part schematically depicted in \autoref{fig:cknPart}.

      The reader may verify that we obtain the black part of the cross diagram
      \begin{equation}
        \begin{tikzpicture}[baseline=(m-2-1.base)]
          \matrix[matrix of math nodes, text height=1.25ex, text depth=0.25ex, ampersand replacement=\&] (m) {
                     \& u \& \\
            c_{k, r, N}^{\textcolor{gray}{-1}} \&   \& |[align=left]| $\begin{cases}
              r_0^{\textcolor{gray}{-1}} & \text{if } u = \langle \mathcal{A} \rangle \text{ such that } \mathcal{A} \text{ satisfies } C_k\\
              \id & \text{otherwise}
            \end{cases}$ \\
                     \& u \& \\
          };
          \foreach \j in {1} {
            \foreach \i in {1} {
              \draw[->] let
                \n1 = {int(2+\i)},
                \n2 = {int(1+\j)}
              in
                (m-\n2-\i) -> (m-\n2-\n1);
              \draw[->] let
                \n1 = {int(1+\i)},
                \n2 = {int(2+\j)}
              in
                (m-\j-\n1) -> (m-\n2-\n1);
            };
          };
        \end{tikzpicture}\label{eqn:ckNCrossDiagram}
      \end{equation}
      and the gray additions with added inverses
      for all $u \in \Sigma^*$ of length $N$, $r \in R$ and all $1 \leq k \leq K$ by construction of the automaton. Note here that the \enquote{otherwise} case occurs if $w$ contains a letter different to $\bot$ and $\top$ (i.\,e.\ it does not encode an assignment) and if $w$ encodes an assignment which does not satisfy $C_k$. Also note that this part is also $\mathcal{R}$-finitary (we are in $\mathcal{R}$ after at most $N$ many letters) and that we may compute it in logarithmic space (since we again only need to count up to $N$).
      
      This concludes the definition of $\mathcal{T}$ and it remains to define $\bm{q}$. For this, recall that we assumed without loss of generality that $K$ is a power of two. 
      To actually define $\bm{q}$, we will use the balanced commutator from \autoref{def:Bcommutator}.
      For this, let $c_k(r) = c_{k, r, N}$ and $c_k(r^{-1}) = c_{k, r, N}^{-1}$ for all $r \in R$ and $1 \leq k \leq K$ and extend it by letting $c_k(\bm{r}) = c_k(r_\ell) \dots c_k(r_1)$ for all $\bm{r} = r_\ell \dots r_1$ with $r_1, \dots, r_\ell \in R^{\pm 1}$ (which is just a substitution of letters and, thus, \LogSpace-computable).
      
      This allows us to define $\bm{q} \in Q^{\pm*}$ as
      \[
        \bm{q} = B \left[ c_K(\bm{r}_K), \dots, c_1(\bm{r}_1) \right] \text{.}
      \]
      Please note that $\bm{q}$ may be computed in logarithmic space by \autoref{fct:BcommutatorIsLogSpaceComputable}.

      This concludes the definition of the reduction function and it remains to show $\bm{q} \neq_\mathcal{T} \id$ if and only if $\varphi$ is satisfiable.
      We start by looking at how $\bm{q}$ acts on a word $u \in \Sigma^*$ of length $N$. From the cross diagrams (\ref{eqn:ckNCrossDiagram}) (and the definition of $c_k(\bm{r})$), we obtain the black part of the cross diagram
      \begin{equation}
        \begin{tikzpicture}[baseline=(q.base), auto]
          \matrix[matrix of math nodes, ampersand replacement=\&,
          text height=1.75ex, text depth=0.25ex] (m) {
                    \& u \& \\
            c_1(\bm{r}_1) \&   \& \bm{q}_1 \\
                    \& u \& \\
             \vdots \& \vdots \& \vdots \\
                    \& u \& \\
            c_K(\bm{r}_K) \&   \& \bm{q}_K \\
                    \& u \& \\
          };
          \foreach \j in {1,5} {
            \foreach \i in {1} {
              \draw[->] let
               \n1 = {int(2+\i)},
                \n2 = {int(1+\j)}
              in
                (m-\n2-\i) -> (m-\n2-\n1);
              \draw[->] let
                \n1 = {int(1+\i)},
                \n2 = {int(2+\j)}
              in
                (m-\j-\n1) -> (m-\n2-\n1);
            };
          };
          
          \node[gray, rotate=90, below=0pt of m-6-1.south, anchor=east, inner sep=0pt] (B0) {$B[$};
          \node[gray, rotate=90, above=0pt of m-2-1.north, anchor=east, inner sep=0pt] (B0c) {$]$};
          \foreach \j in {2,4} {
            \path let
              \n1 = {int(2+\j)}
            in
              node[gray, rotate=90, anchor=base] at ($(m-\j-1)!0.5!(m-\n1-1)$ |- B0.base) {$,$};
          };
          \node[gray, rotate=90, below=0pt of m-6-3.south, anchor=east, inner sep=0pt] (B) {$B[$};
          \node[gray, rotate=90, above=0pt of m-2-3.north, anchor=east, inner sep=0pt] (Bc) {$]$};
          \foreach \j in {2,4} {
            \path let
              \n1 = {int(2+\j)}
            in
              node[gray, rotate=90, anchor=base] at ($(m-\j-3)!0.5!(m-\n1-3)$ |- B.base) {$,$};
          };
          
          \draw[gray, decorate, decoration={brace}] ($(B0.west-|m-6-1.west)+(0pt,0pt)$) -- node (q) {$\bm{q} ={}$} ($(B0c.east-|m-6-1.west)+(0pt,0pt)$);
          
          \draw[gray, decorate, decoration={brace}] ($(Bc.east-|m-6-3.east)+(0pt,0pt)$) -- node {${}= \bm{q}'$} ($(B.west-|m-6-3.east)+(0pt,0pt)$);
        \end{tikzpicture}\label{eqn:qOnW}
      \end{equation}
      where we have $\bm{q}_k = \bm{r}_k$ if $u = \langle \mathcal{A} \rangle$ for some assignment $\mathcal{A}$ that satisfies $C_k$ and that $\bm{q}_k$ consists only of $\id$ (and its inverse) otherwise (i.\,e.\ if $u$ does not encode an assignment or if the assignment does not satisfy $C_k$).
      By \autoref{fct:commutatorInCrossDiagrams}, we may add the balanced commutators to the cross diagram (gray additions above).
      
      We show next that we have
      \[
        \bm{q}' =_{\mathcal{T}} \begin{cases}
          B[ \bm{r}_K, \dots, \bm{r}_1 ] & \text{if } u = \langle \mathcal{A} \rangle \text{ such that $\mathcal{A}$ satisfies $\varphi$} \\
          \id & \text{otherwise}
        \end{cases}
      \]
      for the state sequence $\bm{q}'$ on the right. The upper case immediately follows from cross diagram (\ref{eqn:qOnW}) because $\mathcal{A}$ satisfies all clauses if (and only if) it satisfies $\varphi = \bigwedge_{k = 1}^K C_k$ by definition. The other case has to be split into two subcases.
      If $u = \langle \mathcal{A} \rangle$ for some assignment $\mathcal{A}$ that does not satisfy $\varphi$, there has to be some $1 \leq k \leq K$ such that $\mathcal{A}$ does not satisfy the clause $C_k$. This implies that $\bm{q}_k$ consists only of $\id$ (and its inverse).
      If $u$ does not encode any assignment (i.\,e.\ $u \not\in \{ \bot, \top \}^*$), all $\bm{q}_1, \dots, \bm{q}_K$ will only consist of $\id$ states (or inverses of $\id$, by cross diagram (\ref{eqn:qOnW}).
      Thus, in both cases, there is some $1 \leq k \leq K$ with $\bm{q}_k =_{\mathcal{T}} \id$ and we have $\bm{q}' = B[\bm{q}_K, \dots, \bm{q}_1] =_\mathcal{T} \id$ by \autoref{fct:commutatorAndNegative}.
      
      We may summarize this by stating that we have the cross diagram
      \begin{equation}
        \begin{tikzpicture}[baseline=(m-2-1.base)]
          \matrix[matrix of math nodes, ampersand replacement=\&,
            text height=1.75ex, text depth=0.25ex] (m) {
                \& u \& \\
              \bm{q} \&\& \bm{q}' \begin{cases}
                = B[ \bm{r}_K, \dots, \bm{r}_1 ] & \text{if } u = \langle \mathcal{A} \rangle \text{ such that $\mathcal{A}$ satisfies $\varphi$} \\
                =_{\mathcal{T}} \id & \text{otherwise}
              \end{cases} \\
                \& u \& \\
            };
            \draw[->] (m-2-1) -> (m-2-3);
            \draw[->] (m-1-2) -> (m-3-2);
        \end{tikzpicture}\label{eqn:finalCross}
      \end{equation}
      for all $u \in \Sigma^*$ of length $|u| = N$.
      
      Now, assume that there is some assignment $\mathcal{A}: \mathbb{X} \to \mathbb{B}$ such that $\mathcal{A}$ satisfies $\varphi$.
      By the choice of the $\bm{r}_k$ (see the hypothesis of the theorem), we have $B[\bm{r}_K, \dots, \bm{r}_1] \neq_{\mathcal{R}} \id$ and, as $\mathcal{R}$ is a subautomaton of $\mathcal{T}$, also $B[\bm{r}_K, \dots, \bm{r}_1] \neq_{\mathcal{T}} \id$.
      In particular, there is some witness $v \in \Sigma^*$ with $B[\bm{r}_K, \dots, \bm{r}_1] \circ v \neq v$ (both with respect to the action of $\mathcal{R}$ and the action of $\mathcal{T}$).
      This -- together with diagram (\ref{eqn:finalCross}) -- yields the cross diagram
      \[
        \begin{tikzpicture}[baseline=(m-2-1.base), auto]
          \matrix[matrix of math nodes, ampersand replacement=\&,
          text height=1.75ex, text depth=0.25ex, column 3/.style={anchor=base west}] (m) {
            \& \langle \mathcal{A} \rangle \& \& v \& \\
            \bm{q} \&   \& B[ \bm{r}_K, \dots, \bm{r}_1 ] \& \& {} \\
            \& \langle \mathcal{A} \rangle \& \& \neq v \& \\
          };
          \foreach \j in {1} {
            \foreach \i in {1,3} {
              \draw[->] let
                \n1 = {int(2+\i)},
                \n2 = {int(1+\j)}
              in
                (m-\n2-\i) -> (m-\n2-\n1);
              \draw[->] let
                \n1 = {int(1+\i)},
                \n2 = {int(2+\j)}
              in
                (m-\j-\n1) -> (m-\n2-\n1);
            };
          };
        \end{tikzpicture}
      \]
      (where $\neq v$ is used as a placeholder for a word different to $v$ and the state sequence on the right is omitted because it is not relevant).
      This shows that $\bm{q}$ acts non-trivially on $\langle \mathcal{A} \rangle v$ (with respect to the action of $\mathcal{T}$).
      
      For the other direction, assume that $\varphi$ is not satisfiable. We will show that $\bm{q}$ acts as the identity on all words of length at least $N$ (and, thus, in particular, also on shorter ones).
      Consider a word $uv \in \Sigma^*$ where $u$ is of length $N$. From the cross diagram (\ref{eqn:finalCross}), we obtain the cross diagram
      \[
        \begin{tikzpicture}[baseline=(m-2-5.base), auto]
          \matrix[matrix of math nodes, ampersand replacement=\&,
          text height=1.75ex, text depth=0.25ex] (m) {
                   \& u \& \& v \& \\
            \bm{q} \&                 \&
            \bm{q}' {}=_\mathcal{T} \id \& \& {} \\
                   \& u \& \& v \& \\
          };
          \draw[->] (m-2-1) -> (m-2-3);
          \draw[->] (m-1-2) -> (m-3-2);
          \foreach \j in {1} {
            \foreach \i in {3} {
              \draw[->] let
                \n1 = {int(2+\i)},
                \n2 = {int(1+\j)}
              in
                (m-\n2-\i) -> (m-\n2-\n1);
              \draw[->] let
                \n1 = {int(1+\i)},
                \n2 = {int(2+\j)}
              in
                (m-\j-\n1) -> (m-\n2-\n1);
            };
          };
        \end{tikzpicture}
      \]
      because the case $u = \langle \mathcal{A} \rangle$ for an assignment $\mathcal{A}$ satisfying $\varphi$ cannot occur (as no such assignment exists). This shows $\bm{q} \circ uv = uv$ and, in general, $\bm{q} =_\mathcal{T} \id$ as desired.
    \end{proof}
    \begin{remark}
      The automaton constructed by the reduction for \autoref{thm:coNPHardFamily} has $R (KN + N + 1)$ many states where $R$ is the number of states of $\mathcal{R}$, $N$ is the numbers of variables in the input formula and $K$ is the number of clauses. The $\mathcal{R}$-depth of the automaton is $N$.
    \end{remark}
  
    \paragraph*{The Uniform Word Problem for Finitary Automaton Groups.}
    We may plug in the (finitary) \GAut generating the group $A_5$ from \autoref{ex:A5} for $\mathcal{R}$ into \autoref{thm:coNPHardFamily}. This yields a family of finitary \GAuta (with alphabet size five) whose uniform word problem is $\coNP$-hard. Together with \autoref{prop:finitaryWPisInCoNP}, this yields:
    \begin{corollary}\label{cor:finitaryWPisCoNPComplete}
      The uniform word problem for finitary automaton groups
      \problem{
        a finitary \GAut $\mathcal{T} = (Q, \Sigma, \delta)$ and\newline
        a state sequence $\bm{q} \in Q^{\pm *}$
      }{
        is $\bm{q} = \idGrp$ in $\mathscr{G}(\mathcal{T})$?
      }\noindent
      is \coNP-complete (under many-one $\LogSpace$-reductions).
      This remains true if we fix a set with five elements as the alphabet of the input automaton.
    \end{corollary}
  
    \paragraph*{Finitely Approximable Automata.}
    The next step now is to reduce the alphabet size from five to two. The main idea is to use the automaton generating Grigrochuk's group from \autoref{ex:grigorchuk} and plug it into \autoref{thm:coNPHardFamily} as the automaton $\mathcal{R}$ (which is possible by \autoref{ex:grigorchukCommutator}).
    Unfortunately, Girgorchuk's group is not finitary (it is what is called a \emph{bounded} automaton group instead; see, for example, \cite{waechter2024word}).\footnote{\dots and there is no other fixed suitable finitary automaton with binary alphabet as those automata always generate solvable (finite) groups; see e.\,g.\ \cite{waechter2024word}.}
    However, this is not a problem since we are considering the uniform version of the word problem (where the automaton is part of the input). The main idea is that we may \enquote{unroll} the cycle in the generating automaton sufficiently often to still distinguish non-identity elements from identity ones.
    
    \pagebreak
    We formalize this in the following concept.
    \begin{definition}\label{def:finiteApproximation}
      The \emph{finite approximation} of \emph{depth} $D$ of a \GAut $\mathcal{T} = (Q, \Sigma, \delta)$ is the finitary \GAut $\mathcal{T}' = (Q', \Sigma, \delta')$ with
      \begin{align*}
        Q' &= \{ (q, i) \mid q \in Q, 0 \leq i < D \} \cup \{ \id \} \text{ and}\\
        \delta'
        &= \{ \trans{\id}{a}{a}{\id} \mid a \in \Sigma \} \\
        &\cup \{ \trans{(p, i)}{a}{b}{(q, i + 1)} \mid \trans{p}{a}{b}{q} \in \delta, 0 \leq i < D - 1 \} \\
        &\cup \{ \trans{(p, D - 1)}{a}{b}{\id} \mid \trans{p}{a}{b}{q} \in \delta \} \text{.}
      \end{align*}
      For a state sequence $\bm{q} = q_\ell \dots, q_1$ with $q_1, \dots, q_\ell \in Q^{\pm 1}$, we call $\bm{q}' = (q_\ell, 0) \dots (q_1, 0)$ (with $(q^{-1}, 0) = (q, 0)^{-1}$ for all $q \in Q$) the \emph{projection} of $\bm{q}$ in $\mathcal{T}'$.
    \end{definition}
    The intuitive idea is that we store the number of letters read in the second component of the states and otherwise operate as in the original automaton $\mathcal{T}$. However, after having read $D$ letters, we simply go into the identity state $\id$. This results in an automaton without any cycles except for the self-loops at $\id$ (in fact, the automaton is layered) and, thus, a finitary one (with depth $D$).
    
    The main idea of this construction is, of course, that the projection of a state sequence in the finite approximation of depth $D$ acts on words of length at most $D$ in the same way as the original state sequence.
    \begin{fact}\label{fct:actionInFiniteApproximation}
      Let $\mathcal{T}' = (Q', \Sigma, \delta')$ be the finite approximation of depth $D$ of some \GAut $\mathcal{T} = (Q, \Sigma, \delta)$ and let $\bm{q}' \in Q'^{\pm *}$ be the projection of a state sequence $\bm{q} \in Q^{\pm *}$ in $\mathcal{T}'$. Then, we have the cross diagram
      \begin{center}
        \begin{tikzpicture}[baseline=(m-2-1.base)]
          \matrix[matrix of math nodes, text height=1.25ex, text depth=0.25ex] (m) {
              & w & \\
            \bm{q}' & & \id \\
              & \bm{q} \circ w & \\
          };
          \foreach \i in {1} {
            \draw[->] let
              \n1 = {int(2+\i)}
            in
              (m-2-\i) -> (m-2-\n1);
            \draw[->] let
              \n1 = {int(1+\i)}
            in
              (m-1-\n1) -> (m-3-\n1);
          };
        \end{tikzpicture}%
      \end{center}
      in $\mathcal{T'}$ for all words $w \in \Sigma^*$ of length $|w| = D$. In particular, we have $\bm{q}' \circ w = \bm{q} \circ w$ for all words $w \in \Sigma^*$ of length $|w| \leq D$ (where the action on the left-hand side is with respect to $\mathcal{T}'$ and the one on the right-hand side is with respect to $\mathcal{T}$).
    \end{fact}
    \begin{proof}
      This follows directly from the construction of $\mathcal{T}'$
      as we have the cross diagram
      \begin{center}
        \begin{tikzpicture}[baseline=(m-4-1.east)]
          \matrix[matrix of math nodes, text height=1.25ex, text depth=0.25ex, ampersand replacement=\&] (m) {
              \& a_1 \& \& \dots \& \& a_D \& \\
            q_0 \& \& q_1 \& \dots \& q_{D - 1} \& \& q_D \\
              \& b_1 \& \& \dots \& \& b_D \& \\
          };
          \foreach \j in {1} {
            \foreach \i in {1, 5} {
              \draw[->] let
                \n1 = {int(2+\i)},
                \n2 = {int(1+\j)}
              in
               (m-\n2-\i) -> (m-\n2-\n1);
              \draw[->] let
                \n1 = {int(1+\i)},
                \n2 = {int(2+\j)}
              in
               (m-\j-\n1) -> (m-\n2-\n1);
            };
          };
        \end{tikzpicture} in $\mathcal{T}$
      \end{center}
      (for $q_1, \dots, q_D \in Q$, $a_1, \dots, a_D, b_1, \dots, b_D \in \Sigma$) if and only if we have the cross diagram
      \begin{center}
        \hspace*{\fill}\begin{tikzpicture}[baseline=(m-4-1.east)]
          \matrix[matrix of math nodes, text height=1.25ex, text depth=0.25ex, ampersand replacement=\&] (m) {
              \& a_1 \& \& \dots \& \& a_D \& \\
            (q_0, 0) \& \& (q_1, 1) \& \dots \& (q_{D - 1}, D - 1) \& \& \id \\
              \& b_1 \& \& \dots \& \& b_D \& \\
          };
          \foreach \j in {1} {
            \foreach \i in {1, 5} {
              \draw[->] let
                \n1 = {int(2+\i)},
                \n2 = {int(1+\j)}
              in
               (m-\n2-\i) -> (m-\n2-\n1);
              \draw[->] let
                \n1 = {int(1+\i)},
                \n2 = {int(2+\j)}
              in
               (m-\j-\n1) -> (m-\n2-\n1);
            };
          };
        \end{tikzpicture} in $\mathcal{T}'$.\qedhere
      \end{center}
    \end{proof}
    
    For appropriate choices of $D$, the finite approximation of a \GAut can be computed efficiently.
    \begin{fact}\label{fct:finiteApproximationLogSpace}
      The finite approximation of depth $D$ of a \GAut can be computed in logarithmic space if $D$ is given in unary. In other words, the function
      \function
        {a \GAut $\mathcal{T}$ and\newline
         a natural number $D$ given in unary}
        {the finite approximation of depth $D$ of $\mathcal{T}$}\noindent
      is \LogSpace-computable.
    \end{fact}
    \begin{proof}
      To compute the construction, we only need to count up to the value of $D - 1$ for the second component and, using a binary representation, this requires space $\log D$.
    \end{proof}
    
    Next, we extend the idea of finitely approximating a single automaton to a whole family of automata.
    \begin{definition}\label{def:finitelyApproximable}
      A family $\mathcal{F}$ of \GAuta is \emph{finitely approximable} if the function
      \function
        [the family $\mathcal{F}$ of \GAuta]
        {a \GAut $\mathcal{T} = (Q, \Sigma, \delta) \in \mathcal{F}$ and\newline
         a state sequence $\bm{q} \in Q^{\pm*}$}
        {a finitary \GAut $\mathcal{T}' = (Q', \Sigma, \delta')$ and\newline
         a state sequence $\bm{q}' \in Q'^{\pm *}$ with $\bm{q}' \neq \idGrp$ in $\mathscr{G}(\mathcal{T}') \iff \bm{q} \neq \idGrp$ in $\mathscr{G}(\mathcal{T})$
        }\noindent
      is $\LogSpace$-computable.
    \end{definition}
    
    The main point of our proof now is that, for contracting automata $\mathcal{R}$, a family of $\mathcal{R}$-finitary automata is finitely approximable.
    \begin{proposition}\label{prop:finitelyApproximable}
      Let $\mathcal{R}$ be a contracting \GAut. Then any family $\mathcal{F}$ of $\mathcal{R}$-finitary \GAuta is finitely approximable.
    \end{proposition}
    \begin{proof}
      Let $\mathcal{R} = (R, \Sigma, \varrho)$.
      Since $\mathcal{R}$ is contracting, there are constants $A$ and $B$ by \autoref{lem:finitarilyContracting} such that, for every $\mathcal{R}$-finitary \GAut $\mathcal{T} = (Q, \Sigma, \delta)$, we have
      \[
        \bm{q} \neq \idGrp \text{ in } \mathscr{G}(\mathcal{T}) \implies
        \exists w \in \Sigma^* : \bm{q} \circ w \neq w \text{ and } |w| \leq |Q| + A \log |\bm{q}| + B
      \]
      for all $\bm{q} \in Q^{\pm *}$.
      
      Now, consider some fixed \GAut $\mathcal{T} = (Q, \Sigma, \delta) \in \mathcal{F}$ (which needs to be $\mathcal{R}$-finitary) and a state sequence $\bm{q} \in Q^{\pm*}$. We need to compute in logarithmic space a finitary \GAut $\mathcal{T}' = (Q', \Sigma, \delta')$ and a state sequence $\bm{q}' \in Q'^{\pm *}$ with $\bm{q}' \neq_{\mathcal{T}'} \varepsilon \iff \bm{q} \neq_\mathcal{T} \varepsilon$.
      
      For the automaton $\mathcal{T}'$, we use the finite approximation of depth $D = |Q| + A \log |\bm{q}| + B$ of $\mathcal{T}$, which can be computed in logarithmic space by \autoref{fct:finiteApproximationLogSpace} (since $|Q|$ and $|\bm{q}|$ are both given in unary by the input). For the state sequence $\bm{q}'$, we may simply choose the projection of $\bm{q}$ in $\mathcal{T}'$, which -- as a simple alphabetic substitution -- is certainly also \LogSpace-computable.
      
      Now, if $\bm{q} \neq \idGrp$ in $\mathscr{G}(\mathcal{T})$, there is (by the above) a witness $w \in \Sigma^*$ with $\bm{q} \circ w \neq w$ (with respect to the action of $\mathcal{T}$) and $|w| \leq |Q| + A \log |\bm{q}| + B = D$. This is also a witness for $\bm{q}' \neq \idGrp$ in $\mathscr{G}(\mathcal{T}')$ as we have $\bm{q}' \circ w = \bm{q} \circ w \neq w$ by \autoref{fct:actionInFiniteApproximation}
      (where the first action is with respect to $\mathcal{T}'$ and the second one is with respect to $\mathcal{T}$).
      Conversely, any witness for $\bm{q}' \neq \idGrp$ in $\mathscr{G}(\mathcal{T}')$ is in particular also a witness for $\bm{q} \neq \idGrp$ in $\mathscr{G}(\mathcal{T}')$.
    \end{proof}
    
    \paragraph*{Binary Alphabet.}
    We have now all the pieces to show the final form of our result:
    \begin{theorem}
      The uniform word problem for finitary automaton groups with binary alphabet
      \problem[
        the binary alphabet $\Sigma = \{ 0, 1 \}$
      ]{
        a finitary \GAut $\mathcal{T} = (Q, \Sigma, \delta)$ and\newline
        a state sequence $\bm{q} \in Q^{\pm *}$
      }{
        is $\bm{q} = \idGrp$ in $\mathscr{G}(\mathcal{T})$?
      }\noindent
      is \coNP-complete.
    \end{theorem}
    \begin{proof}
      As before, the problem is in \coNP by \autoref{prop:finitaryWPisInCoNP}.
      
      Let $\mathcal{G}$ denote the \GAut from \autoref{ex:grigorchuk} generating Grigorchuk's group. By \autoref{ex:grigorchukCommutator}, we obtain from \autoref{thm:coNPHardFamily} that the problem
      \problem{
        a $\mathcal{G}$-finitary \GAut $\mathcal{T} = (Q, \Sigma, \delta)$ and\newline
        a state sequence $\bm{q} \in Q^{\pm *}$
      }{
        is $\bm{q} = \idGrp$ in $\mathscr{G}(\mathcal{T})$?
      }\noindent
      is \coNP-hard. Recall from \autoref{ex:grigorchuk} that $\mathcal{G}$ is contracting. Thus, the family of $\mathcal{G}$-finitary \GAuta is finitely approximable by \autoref{prop:finitelyApproximable} and the function from \autoref{def:finitelyApproximable} yields a \LogSpace-reduction from the above problem to the problem in the theorem statement.
    \end{proof}
  \end{section}

  \begin{section}{The Compressed Word Problem}
    \paragraph*{Straight-Line Programs.}
    A \emph{straight-line program} (or \emph{SPL}) is a context-free grammar which generates exactly one word. A \emph{context-free grammar} mainly consists of a set of \emph{rules} where the left-hand side consists of a single \emph{non-terminal symbol} and the right-hand side is a finite word whose letters may be non-terminal or \emph{terminal symbols}. A word is generated by starting at a dedicated non-terminal \emph{starting symbol} and then iteratively replacing non-terminal symbols by matching right-hand sides of rules until only terminal symbols are left. By convention, non-terminal symbols are usually capitalized while terminal symbols are lowercase.
    More details may be found in any introductory textbook on formal language theory (see e.\,g.\ \cite{hopcroft1979introduction}).
    \begin{remark}
      The word generated by an SLP may be exponential in the size of the SLP. An example for this is given by the rules $A_1 \to a$ and $A_{n + 1} \to A_n A_n$ (for $1 \leq n \leq N$).
    \end{remark}
  
    \paragraph*{The Compressed Word Problem.}
    The uniform compressed word problem for finitary automaton groups is the problem
    \problem{
      a finitary \GAut $\mathcal{T} = (Q, \Sigma, \delta)$ and\newline
      a straight-line program generating a state sequence $\bm{q} \in Q^{\pm *}$
    }{
      is $\bm{q} = \idGrp$ in $\mathscr{G}(\mathcal{T})$?
    }\noindent
    The difference to its ordinary version is that the state sequence is not given directly but only by a generating straight-line program. Due to the potential exponential blow-up when decompressing the SLP, the complexity of the compressed version differs in many cases from the one of the ordinary word problem. More information on the compressed word problem may be found in \cite{bassino2020complexity, lohrey2014compressed}.
    
    In this section, we will show that the uniform compressed word problem for finitary automaton groups is \PSpace-complete. We will first do this by giving a direct reduction from the satisfiability problem for quantified boolean formulae but later we give another simpler but less direct proof by finitely approximating the compressed word problem of Grigorchuk's group.
    
    First, however, we prove the easier direction and describe how the uniform word problem for finitary automaton groups can be solved in polynomial space.
    \begin{proposition}\label{prop:compressedWPisInPSpace}
      The uniform compressed word problem for finitary automaton groups is in \PSpace.
    \end{proposition}
    \begin{proof}
      We follow the same guess and check approach as in the proof for \autoref{prop:finitaryWPisInCoNP}. Since the length of the witness (on which $\bm{q}$ acts non-trivially) is bounded by the size of $\mathcal{T}$, it can clearly be guessed in linear space. The more interesting part is the \enquote{check} part. Here, we cannot simply decompress $\bm{q}$ and then apply it state by state (since $\bm{q}$ can be exponentially long). However, we can still compute (and store) the intermediate $u_i$ directly from the SLP. We start with the rule $S \to \alpha_\ell \dots \alpha_1$ where the $\alpha_i$ are either terminal symbols (i.\,e.\ states) or non-terminals. We apply the symbols $\alpha_i$ from right to left to $u$. If $\alpha_i$ is a state, we can directly apply it to the current word. If it is a non-terminal symbol $\alpha_i = B$, we descend recursively into the rule $B \to \beta_k \dots \beta_1$. For this, we have to store where we were in the previous rule (this can, for example, be done using a pointer, which is clearly possible even in linear space). Note that we may assume that the same non-terminal symbol does not appear twice in the same recursive branch as this would correspond to a syntax tree with multiple instances of the same non-terminal symbol on one branch, which cannot occur if the grammar only generates a single word. Thus, in the worst case, we need to store one position for every rule in the input, which is still possible in linear space.
    \end{proof}
  
    \paragraph*{A Reduction from Quantified Boolean Formulae.}
    For the other direction -- namely to prove that our problem is \PSpace-hard -- we use the following problem for the reduction.
    \begin{theorem}\label{thm:3QBF}
      The problem \DecProblem{$3$-QBF}
      \problem{
        a quantified boolean formula $\varphi = \lnot \forall x_N \lnot \forall x_{N - 1} \dots \lnot \forall x_1: \varphi_0$\newline
        where $\varphi_0$ is in 3-conjunctive normal form and\newline
        contains no variables other than $\{ x_1, \dots, x_N \}$.
      }{
        is $\varphi$ true?
      }\noindent
      is \PSpace-complete (under many-one \LogSpace-reductions).
    \end{theorem}
    \begin{proof}
      \enlargethispage{2\baselineskip}
      We reduce the problem
      \problem{
        $\varphi = \exists x_1 \forall x_2 \dots Q_N x_N: \varphi_0$ where $\varphi_0$ is in conjunctive normal form
      }{
        does $\varphi$ hold?
      }\noindent
      where $Q_N = \exists$ if $N$ is odd and $Q_N = \forall$ is $N$ is even,      which is \PSpace-complete (under many-one \LogSpace-reductions) by \cite[Theorem~19.1]{papadimitriou97computational}, to the special version stated in the theorem.
      
      First, we split up all clauses with more than three literals in the common way\footnote{This is usually used to prove that \DecProblem{$3$-SAT} is \NP-complete (see e.\,g.\ \cite[Problem~9.5.2]{papadimitriou97computational}).} by using the fact that $L_1 \lor \dots \lor L_\ell$ and $\exists z: \left( (L_1 \lor L_2 \lor z) \land (\lnot z \lor L_3 \lor \dots L_\ell) \right)$ (where $z$ is a new, so-far unused variable) are equivalent (i.\,e.\ they are satisfied by exactly the same assignments). This introduces additional existential quantifiers at the innermost position.
      
      Clauses with less than three literals can be padded with new variables by using the fact that any literal $L$ is equivalent to $\forall z: (L \lor z)$ (where $z$ is again a new variable). This introduces additional universal quantifiers at the innermost position.
   
      We may ensure that the quantifiers are alternating between $\exists$ and $\forall$ by adding dummy variables not appearing in the matrix (i.\,e.\ the inner part of the formula without quantifiers) of the formula. This results in a formula of the form $\exists x_{N'} \forall x_{N' - 1} \dots \exists x_2 \forall x_1: \varphi_0'$ where $\varphi_0'$ is in $3$-conjunctive normal form which is equivalent to the original formula.
    
      Finally, we use the equivalence of $\exists z: \psi$ and $\lnot \forall z: \lnot \psi$ to eliminate all existential quantifiers.
      
      Note that each of these steps can be computed in \LogSpace and that, thus, the whole reduction can be done in \LogSpace.
    \end{proof}
  
    \begin{proposition}\label{prop:compressedWPisPSpaceHard}
      The uniform compressed word problem for finitary automaton groups is \PSpace-hard (under many-one \LogSpace-reductions).
      This remains true if we fix a set with five elements as the alphabet of the input automaton.
    \end{proposition}
    \begin{proof}
      We reduce \DecProblem{$3$-QBF} from \autoref{thm:3QBF} to the (complement of the) compressed word problem for fintary automaton groups (in $\LogSpace$). For this, assume that we get a quantified boolean formula $\varphi = \lnot \forall x_N \lnot \forall x_{N - 1} \dots \lnot \forall x_1: \varphi_0$ where $\varphi_0$ is in 3-conjunctive normal form and contains no variables other than $\{ x_1, \dots, x_N \}$.
      
      As a first step, we use the reduction described in \autoref{thm:coNPHardFamily} for the matrix $\varphi_0$ of our input formula and the automaton generating $A_5$ from \autoref{ex:A5} as $\mathcal{R}$. This yields a finitary \GAut $\mathcal{T}_0 = (Q_0, \Sigma, \delta_0)$ with $\Sigma = \{ a_1, \dots, a_5 \}$ and a state sequence $\bm{q}_0 \in Q_0^{\pm *}$ with
      \begin{center}
        \begin{tikzpicture}[baseline=(m-2-1.base)]
          \matrix[matrix of math nodes, ampersand replacement=\&,
            text height=1.75ex, text depth=0.25ex] (m) {
                \& u \& \\
              \bm{q}_0 \&\& \tilde{\bm{q}}_0 \begin{cases}
                = B[ \bm{r}_K, \dots, \bm{r}_1 ] & \text{if } u = \langle \mathcal{A} \rangle \text{ such that $\mathcal{A}$ satisfies $\varphi_0$} \\
                =_{\mathcal{T}_0} \id & \text{otherwise}
              \end{cases} \\
                \& u \& \\
            };
            \draw[->] (m-2-1) -> (m-2-3);
            \draw[->] (m-1-2) -> (m-3-2);
        \end{tikzpicture}
      \end{center}
      for all $u \in \Sigma^*$ of length $|u| = N$ by cross diagram (\ref{eqn:finalCross}) where we have re-used the notation $\langle \mathcal{A} \rangle$ from the proof of \autoref{thm:coNPHardFamily}. With our special choice of using the automaton generating $A_5$ from \autoref{ex:A5} for $\mathcal{R}$, we have $B[ \bm{r}_K, \dots, \bm{r}_1 ] =_{\mathcal{R}} \sigma$ and, since $\mathcal{R}$ is a subautomaton of $\mathcal{T}_0$, we even have $B[ \bm{r}_K, \dots, \bm{r}_1 ] =_{\mathcal{T}_0} \sigma$. This simplifies the above cross diagram into the diagram
      \begin{equation}
        \begin{tikzpicture}[baseline=(m-2-1.base)]
          \matrix[matrix of math nodes, ampersand replacement=\&,
            text height=1.75ex, text depth=0.25ex] (m) {
                \& u \& \\
              \bm{q}_0 \&\& =_{\mathcal{T}_0} \begin{cases}
                \sigma & \text{if } u = \langle \mathcal{A} \rangle \text{ such that $\mathcal{A}$ satisfies $\varphi_0$} \\
                \id & \text{otherwise}
              \end{cases} \\
                \& u \& \\
            };
            \draw[->] (m-2-1) -> (m-2-3);
            \draw[->] (m-1-2) -> (m-3-2);
        \end{tikzpicture}\label{eqn:simplifiedCross}
      \end{equation}
      for $u \in \Sigma^*$ of length $|u| = N$.
      Recall that the \enquote{otherwise} case occurs in two cases:
      first, if $u$ does not encode any assignment (i.\,e.\ if $u \not\in \{ \bot, \top \}$ for the special elements $\bot$ and $\top$ from $\Sigma$ chosen in the proof of \autoref{thm:coNPHardFamily}) and, second, if $u = \langle \mathcal{A} \rangle$ but the assignment $\mathcal{A}$ does not satisfy $\varphi_0$.
      
      Another thing to recall from the proof of \autoref{thm:coNPHardFamily} (see cross diagram (\ref{eqn:sigmaCrossDiagram}) and \autoref{fig:sigmaPart}) is that $\mathcal{T}_0$ contains (in particular) the states
      $
        \{ \sigma_n, \alpha_n, \beta_n \mid 0 < n \leq N \}
      $
      with the cross diagram
      \begin{equation}
        \begin{tikzpicture}[baseline=(m-2-1.base)]
          \matrix[matrix of math nodes, text height=1.25ex, text depth=0.25ex, ampersand replacement=\&] (m) {
                     \& w \& \\
            \gamma_n \&   \& |[align=left]| $\begin{cases}
              \gamma & \text{if } w \in \{ \bot, \top \}^* \\
              \id & \text{otherwise}
            \end{cases}$ \\
                     \& w \& \\
          };
          \foreach \j in {1} {
            \foreach \i in {1} {
              \draw[->] let
                \n1 = {int(2+\i)},
                \n2 = {int(1+\j)}
              in
                (m-\n2-\i) -> (m-\n2-\n1);
              \draw[->] let
                \n1 = {int(1+\i)},
                \n2 = {int(2+\j)}
              in
                (m-\j-\n1) -> (m-\n2-\n1);
            };
          };
        \end{tikzpicture}\label{eqn:gammaN}
      \end{equation}
      for $\gamma \in \{ \sigma, \alpha, \beta \}$, $0 \leq n \leq N$ and all words $w \in \Sigma^*$ of length $|w| = n$
      (where $\gamma$ is a state in the automaton generating $A_5$ from \autoref{ex:A5}).
      
      We perform a second reduction on the output of the first one (which is possible since $\LogSpace$-computable functions are closed under composition, see, for example, \cite[Proposition~8.2]{papadimitriou97computational}).
      Here, we need to compute a finitary \GAut $\mathcal{T}$ and a state sequence $\bm{q}$ encoded as an SLP such that $\bm{q} \neq_\mathcal{T} \id$ if and only if $\varphi$ holds.
      
      \begin{figure}\centering
        \begin{tikzpicture}[auto, shorten >=1pt, >=latex]
          \node[state, ellipse, inner sep=0pt] (N) {$t_{N-1}$};
          \node[state, ellipse, inner sep=0pt, right=of N] (N-1) {$t_{N - 2}$};
          \node[right=of N-1] (dots) {$\dots$};
          \node[state, right=of dots] (0) {$t_{0}$};
          \node[state, right=of 0, dashed] (id) {$\id$};
          
          \draw[->] (N) edge node[above] {$\bot / \bot$} node[below] {$\top / \top$} (N-1)
                    (N-1) edge node[above] {$\bot / \bot$} node[below] {$\top / \top$} (dots)
                    (dots) edge node[above] {$\bot / \bot$} node[below] {$\top / \top$} (0)
                    (0) edge[thick] node[above] {$\bm{\bot / \top}$} node[below] {$\bm{\top / \bot}$} (id);
        \end{tikzpicture}
        \caption{The additional automaton part with the states $\{ t_n \mid 0 \leq n < N \}$. Missing transitions are $b/b$ transitions to the identity state (for $b \in \Sigma \setminus \{ \bot, \top \}$).}\label{fig:tPart}
      \end{figure}
      To obtain the automaton $\mathcal{T} = (Q, \Sigma, \delta)$, we extend $\mathcal{T}_0$ by some additional states (but keep the alphabet the same: $\Sigma = \{ a_1, \dots, a_5 \}$). The new states are $\{ t_n \mid 0 \leq n < N \}$ with the additional transitions
      \begin{align*}
        &\left\{ \trans{t_n}{\bot}{\bot}{t_{n - 1}}, \trans{t_n}{\top}{\top}{t_{n - 1}}, \trans{t_n}{b}{b}{\id} \mid 0 < n < N, b \in \Sigma \setminus \{ \bot, \top \} \right\}\\
        {}\cup{}& \left\{ \trans{t_0}{\bot}{\top}{\id}, \trans{t_0}{\top}{\bot}{\id}, \trans{t_0}{b}{b}{\id} \mid b \in \Sigma \setminus \{ \bot, \top \} \right\}
        \subseteq \delta \text{.}
      \end{align*}
      This new automaton part is depicted in \autoref{fig:tPart}. Note that we have not introduced any cycles and that this new part may be computed in logarithmic space (as we only need a counter up to the value of $N$).
      By construction, we obtain the cross diagram
      \begin{equation}
        \begin{tikzpicture}[baseline=(m-2-1.base), auto]
          \matrix[matrix of math nodes, ampersand replacement=\&,
            text height=1.75ex, text depth=0.25ex] (m) {
            \& w \& \& a \& \\
            t_n \& \& t_0 \& \& \id \\
            \& w \& \& \lnot a \& \\
          };
        
          \foreach \j in {1} {
            \foreach \i in {1,3} {
              \draw[->] let
                \n1 = {int(2+\i)},
                \n2 = {int(1+\j)}
              in
                (m-\n2-\i) -> (m-\n2-\n1);
              \draw[->] let
                \n1 = {int(1+\i)},
                \n2 = {int(2+\j)}
              in
                (m-\j-\n1) -> (m-\n2-\n1);
            };
          };
        \end{tikzpicture}
      \end{equation}
      for all $0 < n < N$ where $w \in \{ \bot, \top \}^*$ is of length $n$, $a \in \{ \bot, \top \}$ and $\lnot a$ denotes the negation of $a$ (i.\,e.\ $\lnot a = \top$ if $a = \bot$ and $\lnot a = \bot$ if $a = \top$). For general words $w \in \Sigma^*$ of length $0< n < N$ and letters $a \in \Sigma$, we get the cross diagram
      \begin{equation}
        \begin{tikzpicture}[baseline=(m-2-1.base), auto]
          \matrix[matrix of math nodes, ampersand replacement=\&,
          text height=1.75ex, text depth=0.25ex] (m) {
            \& w \& \& a \& \\
            t_n \& \& t' \& \& \id \\
            \& w \& \& \tilde{a} \& \\
          };
          
          \foreach \j in {1} {
            \foreach \i in {1,3} {
              \draw[->] let
                \n1 = {int(2+\i)},
                \n2 = {int(1+\j)}
              in
               (m-\n2-\i) -> (m-\n2-\n1);
              \draw[->] let
                \n1 = {int(1+\i)},
                \n2 = {int(2+\j)}
              in
                (m-\j-\n1) -> (m-\n2-\n1);
            };
          };
        \end{tikzpicture}\label{eqn:tnCrossDiagram}
      \end{equation}
      where we have $t' = t_0$ and $\tilde{a} = \lnot a$ if $wa \in \{ \bot, \top \}^*$ and $t' = \id$ and $\tilde{a} = a$ otherwise.
      
      We will define the state sequence $\bm{q}$ inductively and will use this inductive structure in the end to compute an SLP generating $\bm{q}$. We already have $\bm{q}_0$ and, for $0 < n \leq N$ and let:
      \[
        \bm{q}_n' = \left[ \bm{q}_{n - 1}^{t_{N - n} \beta_N}, \bm{q}_{n - 1}^{\alpha_N} \right] \text{ and }
        \bm{q}_n = \left( \bm{q}_n' \right)^{-1} \sigma_N
      \]
      
      The reason for choosing the $\bm{q}_n$ in this way is to satisfy a certain invariant. To state it, recall that $\varphi_0$ is already given, let $\varphi_0' = \varphi_0$ and
      \[
        \varphi_n' = \forall x_n: \varphi_{n - 1} \text{ and } \varphi_n = \lnot \varphi_n'
      \]
      for $0 < n \leq N$. Note that this means
      \[
        \varphi_n = \lnot \forall x_n \dots \lnot \forall x_1: \varphi_0
      \]
      and $\varphi_n'$ is the same except that it misses the out-most negation. In particular, we have $\varphi_N = \varphi$.
      Before we finally state the invariant, we extend the notation $\langle \mathcal{A} \rangle$ to assignments $\mathcal{A}: \{ x_{n + 1}, \dots, x_N \} \to \mathbb{B}$ (for $0 \leq n \leq N$) by letting $\langle \mathcal{A} \rangle = \mathcal{A}(x_N) \dots \mathcal{A}(x_{n + 1}) \in \{ \bot, \top \}^* \subseteq \Sigma^*$. Note that $|\langle \mathcal{A} \rangle|$ has length $N - n$ and that the empty word is the encoding of an empty assignment. Now, the invariant we want to satisfy with the $\bm{q}_n$ is that, for all $0 \leq n \leq N$, all words $u \in \Sigma^*$ of length $N - n$ and all words $v \in \Sigma^*$ of length $n$, we have the black part of the cross diagram
      \begin{equation}
        \begin{tikzpicture}[baseline=(q.base), auto]
          \matrix[matrix of math nodes, ampersand replacement=\&,
            text height=1.75ex, text depth=0.25ex, column 5/.style={anchor=base west}] (m) {
            \& |[gray]| u \& \& |[gray]| v \& \\
            |[gray]| \sigma_N \& \& {} \& \& |[gray]| = \begin{cases}
              \sigma & \text{if } uv \in \{ \bot, \top \}^* \\
              \id & \text{otherwise}
            \end{cases} \\
            \& u \& \& v \& \\
            \textcolor{gray}{(} \bm{q}_n' \textcolor{gray}{)^{-1}} \&   \& {} \& \& =_\mathcal{T}
              \begin{cases}
                \textcolor{gray}{(} \sigma \textcolor{gray}{)^{-1}} & \text{if } u = \langle \mathcal{A} \rangle \text{ s.\,t.\ } \mathcal{A} \text{ satisfies } \varphi_n' \text{ and } v \in \{ \bot, \top \}^* \\
                \id & \text{otherwise}
              \end{cases}
             \\
            \& u \& \& v \& \\
          };
          \foreach \j in {1} {
            \foreach \i in {1,3} {
              \draw[->, gray] let
                \n1 = {int(2+\i)},
                \n2 = {int(1+\j)}
              in
                (m-\n2-\i) -> (m-\n2-\n1);
              \draw[->, gray] let
                \n1 = {int(1+\i)},
                \n2 = {int(2+\j)}
              in
                (m-\j-\n1) -> (m-\n2-\n1);
            };
          };
        
          \foreach \j in {3} {
            \foreach \i in {1,3} {
              \draw[->] let
                \n1 = {int(2+\i)},
                \n2 = {int(1+\j)}
              in
               (m-\n2-\i) -> (m-\n2-\n1);
              \draw[->] let
                \n1 = {int(1+\i)},
                \n2 = {int(2+\j)}
              in
                (m-\j-\n1) -> (m-\n2-\n1);
            };
          };
        
          \coordinate (topLeft) at (m-2-1.north west -| m-4-1.south west);
          \draw[gray, decorate, decoration={brace}] ($(m-4-1.south west)+(0pt,0pt)$) -- node (q) {$\bm{q}_n ={}$} ($(topLeft)+(0pt,0pt)$);
        \end{tikzpicture}\label{eqn:invariant}
      \end{equation}
      where we let $\bm{q}_0' = \bm{q}_0$ and use the convention that the empty assignment satisfies a (closed)\footnote{A formula is \emph{closed} if it does not have any free variables, i.\,e.\ if all appearing variables are bound by a quantifier.} formula if and only if the formula holds. Note that the (black) \enquote{otherwise} case includes the case that $u$ or $v$ is not from $\{ \bot, \top \}^*$ and the case that $u$ encodes an assignment not satisfying $\varphi_n'$.
      
      As soon as we have established this invariant for some $n$, we immediately also get a version where we take the inverses of the states (this is possible since the action is trivial; normally, we would have to additionally flip the diagram along the horizontal axis).
      Using the cross diagram (\ref{eqn:gammaN}), we may add an additional line for $\sigma_N$ and obtain the gray additions to the above diagram for $0 < n \leq N$. Note that the product of the state sequences on the right hand side acts trivially if $u = \langle \mathcal{A} \rangle$ for some $\mathcal{A}$ which satisfies $\varphi_n'$ (this is the case if and only if $\mathcal{A}$ does not satisfy $\varphi_n = \lnot \varphi_n'$) and $v \in \{ \bot, \top \}^*$. It also acts trivially if $uv \not\in \{ \bot, \top \}^*$. On the other hand, it acts like $\sigma$ if $u = \langle \mathcal{A} \rangle$ for some $\mathcal{A}$ which does satisfy $\varphi_n = \lnot \varphi_n'$ and $v \in \{ \bot, \top \}^*$. This yields the cross diagram
      \begin{equation}
        \begin{tikzpicture}[baseline=(q.base), auto]
          \matrix[matrix of math nodes, ampersand replacement=\&,
            text height=1.75ex, text depth=0.25ex, column 5/.style={anchor=base west}] (m) {
            \& u \& \& v \& \\
            \bm{q}_n \& \& {} \& \& =_\mathcal{T}
              \begin{cases}
                \sigma & \text{if } u = \langle \mathcal{A} \rangle \text{ s.\,t.\ } \mathcal{A} \text{ satisfies } \varphi_n \text{ and } v \in \{ \bot, \top \}^* \\
                \id & \text{otherwise}
              \end{cases}\\
            \& u \& \& v \& \\
          };
          \foreach \j in {1} {
            \foreach \i in {1,3} {
              \draw[->] let
                \n1 = {int(2+\i)},
                \n2 = {int(1+\j)}
              in
               (m-\n2-\i) -> (m-\n2-\n1);
              \draw[->] let
                \n1 = {int(1+\i)},
                \n2 = {int(2+\j)}
              in
                (m-\j-\n1) -> (m-\n2-\n1);
            };
          };
        \end{tikzpicture}\label{eqn:qnInvariant}
      \end{equation}
      for all $0 \leq n \leq N$, all words $u \in \Sigma^*$ for length $N - n$ and all words $v \in \Sigma^*$ of length $n$.
      
      To prove the invariant (i.\,e.\ the black part of cross diagram (\ref{eqn:invariant})), we use induction on $n$. For $n = 0$, we have to show the cross diagram
      \[
        \begin{tikzpicture}[baseline=(q.base), auto]
          \matrix[matrix of math nodes, ampersand replacement=\&,
            text height=1.75ex, text depth=0.25ex, column 5/.style={anchor=base west}] (m) {
            \& u \& \& \varepsilon \& \\
            \bm{q}_0' = \bm{q}_0 \& \& {} \& \& =_\mathcal{T}
              \begin{cases}
                \sigma & \text{if } u = \langle \mathcal{A} \rangle \text{ s.\,t.\ } \mathcal{A} \text{ satisfies } \varphi_0' \text{ (which is equal to } \varphi_0 \text{)} \\
                \id & \text{otherwise}
              \end{cases}\\
            \& u \& \& \varepsilon \& \\
          };
          \foreach \j in {1} {
            \foreach \i in {1,3} {
              \draw[->] let
                \n1 = {int(2+\i)},
                \n2 = {int(1+\j)}
              in
               (m-\n2-\i) -> (m-\n2-\n1);
              \draw[->] let
                \n1 = {int(1+\i)},
                \n2 = {int(2+\j)}
              in
                (m-\j-\n1) -> (m-\n2-\n1);
            };
          };
        \end{tikzpicture}
      \]
      for $u \in \Sigma^*$ of length $N$. This, however, is exactly cross diagram (\ref{eqn:simplifiedCross}) (when we observe that $\mathcal{T}_0$ is a subautomaton of $\mathcal{T}$)
      
      \begin{figure}\centering
        \begin{tikzpicture}[baseline=(q.base), auto]
          \matrix[matrix of math nodes, ampersand replacement=\&, text height=1.75ex, text depth=0.25ex,
          ] (m) {
              \& u \& \& a \& \& v \& \\
            \alpha_N \& \& {} \& \& {} \& \& x \\
              \& u \& \& a \& \& v \& \\
            \bm{q}_{n - 1} \& \& {} \& \& {} \& \& \bm{p}_{n, 0} \\
              \& u \& \& a \& \& v \& \\
            \alpha_N^{-1} \& \& {} \& \& {} \& \& x^{-1} \\
              \& u \& \& a \& \& v \& \\
            \beta_N \& \& {} \& \& {} \& \& y \\
              \& u \& \& a \& \& v \& \\
            t_{N - n} \& \& {} \& \& \id \& \& \id \\
              \& u \& \& \tilde{a} \& \& v \& \\
            \bm{q}_{n - 1} \& \& {} \& \& {} \& \& \bm{p}_{n, 1} \\
              \& u \& \& \tilde{a} \& \& v \& \\
            t_{N - n}^{-1} \& \& {} \& \& \id \& \& \id \\
              \& u \& \& a \& \& v \& \\
            \beta_N^{-1} \& \& {} \& \& {} \& \& y^{-1} \\
              \& u \& \& a \& \& v \& \\
          };
          \foreach \j in {1,3,5,7,9,11,13,15} {
            \foreach \i in {1,3,5} {
              \draw[->] let
                \n1 = {int(2+\i)},
                \n2 = {int(1+\j)}
              in
               (m-\n2-\i) -> (m-\n2-\n1);
              \draw[->] let
                \n1 = {int(1+\i)},
                \n2 = {int(2+\j)}
              in
                (m-\j-\n1) -> (m-\n2-\n1);
            };
          };

          \path[fill=gray, opacity=0.2, rounded corners] (m-4-1.west |- m-3-2.north) rectangle (m-4-7.east |- m-5-6.south);
          \path[fill=gray, opacity=0.2, rounded corners] (m-12-1.west |- m-11-2.north) rectangle (m-12-7.east |- m-13-6.south);
          
          \node[gray, rotate=90, below=0pt of m-16-1.south, anchor=east, inner sep=0pt] (BN) {$\Big[$};
          \node[gray, rotate=90, above=0pt of m-2-1.north, anchor=east, inner sep=0pt] (BNc) {$\Big]$};
          \node[gray, rotate=90, anchor=base] at ($(m-6-1)!0.5!(m-8-1)$ |- BN.base) {$,$};
          
          \draw[gray, decorate, decoration={brace}] ($(BN.north west)+(-8pt,0pt)$) -- node (q) {$\bm{q}'_n ={}$} ($(BNc.north east)+(-8pt,0pt)$);
          
          \node[gray, rotate=90, below=0pt of m-16-7.south, anchor=east, inner sep=0pt] (B0) {$\Big[$};
          \node[gray, rotate=90, above=0pt of m-2-7.north, anchor=east, inner sep=0pt] (B0c) {$\Big]$};
          \node[gray, rotate=90, anchor=base] at ($(m-6-7)!0.5!(m-8-7)$ |- B0.base) {$,$};
        \end{tikzpicture}
        \caption{Cross diagram for the inductive step. The shaded parts are due to induction (compare to cross diagram (\ref{eqn:qnInvariant})), the lines involving $\alpha_N$ or $\beta_N$ are due to cross diagram (\ref{eqn:gammaN}) and the ones involving $t_{N - n}$ follow from cross diagram (\ref{eqn:tnCrossDiagram}). The commutator may be added (gray additions) due to \autoref{fct:commutatorInCrossDiagrams} (with $d = 1$/$D = 2$).}
        \label{fig:inductiveStepCrossDiagram}
      \end{figure}
      For the inductive step from $n - 1$ to $n$, consider a word $u \in \Sigma^*$ of length $N - n$, $a \in \Sigma$ and $v \in \Sigma^*$ of length $n - 1$. We have the black part of the cross diagram in \autoref{fig:inductiveStepCrossDiagram}
      where we have, by cross diagram (\ref{eqn:gammaN}),
      \begin{align*}
        x &= \begin{cases}
          \alpha & \text{if } uav \in \{ \bot, \top \}^* \\
          \id & \text{otherwise}
        \end{cases} \\
        y &= \begin{cases}
          \beta & \text{if } uav \in \{ \bot, \top \}^* \\
          \id & \text{otherwise}
        \end{cases}
      \shortintertext{and, by induction/cross diagram (\ref{eqn:qnInvariant}),}
        \bm{p}_{n, 0} &=_\mathcal{T} \begin{cases}
         \sigma & \text{if } ua = \langle \mathcal{A}' \rangle \text{ s.\,t.\ } \mathcal{A}' \text{ satisfies } \varphi_{n - 1} \text{ and } v \in \{ \bot, \top \}^* \\
         \id & \text{otherwise}
        \end{cases}\\
        \bm{p}_{n, 1} &=_\mathcal{T} \begin{cases}
          \sigma & \text{if } u\tilde{a} = \langle \tilde{\mathcal{A}}' \rangle \text{ s.\,t.\ } \tilde{\mathcal{A}}' \text{ satisfies } \varphi_{n - 1} \text{ and } v \in \{ \bot, \top \}^* \\
          \id & \text{otherwise.}
        \end{cases}
      \end{align*}
      The shaded parts are due to induction (compare to cross diagram (\ref{eqn:qnInvariant})), the lines involving $\alpha_N$ or $\beta_N$ are due to cross diagram (\ref{eqn:gammaN})
      and the ones involving $t_{N - n}$ follow from cross diagram (\ref{eqn:tnCrossDiagram}).
      We may add the commutators to the diagram (due to \autoref{fct:commutatorInCrossDiagrams}) and obtain the gray additions.
      
      The rest of the inductive step is now a case distinction. If we have $uav \not\in \{ \bot, \top \}^*$, we get $x = \id$, $y = \id$, $\bm{p}_{n, 0} =_\mathcal{T} \bm{p}_{n, 1} =_\mathcal{T} \id$ and, thus, for the state sequence on the right $[ \bm{p}_{n, 1}^{y}, \bm{p}_{n, 0}^{x} ] =_\mathcal{T} \id$ (since it consists only of $\id$ and $\id^{-1} =_{\mathcal{T}} \id$ states).
      
      Now, assume $uav \in \{ \bot, \top \}^*$ and, in particular, $a \in \{ \bot, \top \}$.
      In this case, we have $x = \alpha$ and $y = \beta$ as well as
      $u = \langle \mathcal{A} \rangle$ for some $\mathcal{A}: \{ x_{n + 1}, \dots, x_N \} \to \mathbb{B}$ and $\tilde{a} = \lnot a$ (see cross diagram (\ref{eqn:tnCrossDiagram})). Let $ua = \langle \mathcal{A}' \rangle$ and $u\tilde{a} = \langle \tilde{\mathcal{A}}' \rangle$. Note that we have $\mathcal{A}'(x_n) = a = \lnot \tilde{\mathcal{A}}'(x_n)$ (and $\mathcal{A}'(x_m) = \tilde{\mathcal{A}}'(x_m) = \mathcal{A}(x_m)$ for all $n < m \leq N$). If $\mathcal{A}$ satisfies $\varphi_n' = \forall x_n : \varphi_{n - 1}$, we, therefore, have that $\mathcal{A}'$ and $\tilde{\mathcal{A}}'$ both satisfy $\varphi_{n - 1}$. This yields $\bm{p}_{n, 0} =_\mathcal{T} \bm{p}_{n, 1} =_\mathcal{T} \sigma$ (by the above equalities for $\bm{p}_{n, 0}$ and $\bm{p}_{n, 1}$) and, thus, for the state sequence on the right $[ \bm{p}_{n, 1}^\beta, \bm{p}_{n, 0}^\alpha ] =_\mathcal{T} \sigma$ (by the choice of $\sigma, \alpha$ and $\beta$ in \autoref{ex:A5}).
      On the other hand, if $\mathcal{A}$ does not satisfy $\varphi_n' = \forall x_n : \varphi_{n - 1}$, we must have that $\mathcal{A}'$ or $\tilde{\mathcal{A}}'$ does not satisfy $\varphi_{n - 1}$. In this case, we have $\bm{p}_{n, 0} =_\mathcal{T} \id$ or $\bm{p}_{n, 1} =_\mathcal{T} \id$ and, thus, for the state sequence on the right, $[ \bm{p}_{n, 1}^\beta, \bm{p}_{n, 0}^\alpha ] =_\mathcal{T} \id$ by \autoref{fct:commutatorAndNegative}.
      This shows that the cases for the gray additions to the cross diagram in \autoref{fig:inductiveStepCrossDiagram} reflect exactly the black part of cross diagram (\ref{eqn:invariant}), which shows the invariant.

      Considering the special case $n = N$ for cross diagram (\ref{eqn:qnInvariant}), we have obtain
      \begin{equation*}
        \begin{tikzpicture}[baseline=(q.base), auto]
          \matrix[matrix of math nodes, ampersand replacement=\&,
          text height=1.75ex, text depth=0.25ex, column 5/.style={anchor=base west}] (m) {
            \& \varepsilon \& \& v \& \\
            \bm{q}_N \& \& {} \& \& =_\mathcal{T}
            \begin{cases}
              \sigma & \text{if } \varphi_N \text{ holds and } v \in \{ \bot, \top \}^* \\
              \id & \text{otherwise}
            \end{cases}\\
            \& \varepsilon \& \& v \& \\
          };
          \foreach \j in {1} {
            \foreach \i in {1,3} {
              \draw[->] let
                \n1 = {int(2+\i)},
                \n2 = {int(1+\j)}
              in
                (m-\n2-\i) -> (m-\n2-\n1);
              \draw[->] let
                \n1 = {int(1+\i)},
                \n2 = {int(2+\j)}
              in
                (m-\j-\n1) -> (m-\n2-\n1);
            };
          };
        \end{tikzpicture}
      \end{equation*}
      for all $v \in \Sigma^*$ of length $N$. This shows that we have $\bm{q}_N =_\mathcal{T} \id$ if $\varphi_N$ (which is equal to $\varphi$) does \textbf{not} hold. If it does hold, on the other hand, we have $\bm{q}_N \circ \bot^N a = \bot^N \sigma(a) \neq \bot^N a$ for some $a \in \Sigma$ (since $\sigma$ is not the identity permutation). Thus, we may choose $\bm{q} = \bm{q}_N$ as the sought state sequence and it remains to show how an SLP generating $\bm{q}_N$ can be computed in logarithmic space.
      
      Note that $\bm{q}_0$ is already given and we may, thus, begin with the rule $A_0 \to \bm{q}_0$ and add the rules
      \begin{align*}
        A_n' &\to \beta_N^{-1} t_{N - n}^{-1} A_{n - 1}^{-1} t_{N - n} \beta_N \  \alpha_N^{-1} A_{n - 1}^{-1} \alpha_N \  \beta_N^{-1} t_{N - n}^{-1} A_{n - 1} t_{N - n} \beta_N \  \alpha_N^{-1} A_{n - 1} \alpha_N\\
        &= \left[ A_{n - 1}^{t_{N - n} \alpha_N}, A_{n - 1}^{\beta_N} \right] \quad \text{and}\\
        A_n &\to (A_n')^{-1} \sigma_N
      \end{align*}
      for $1 \leq n \leq N$, where we also implicitly add the rules for $(A_n')^{-1}$ and $A_n^{-1}$ by mirroring the right-hand sides and inverting every symbol. Note that these rules may be computed in logarithmic space. We choose $A_N$ as our starting symbol and the reader may verify that $A_n'$ generates $\bm{q}_n'$ and $A_n$ generates $\bm{q}_n$ (this follows directly from the inductive definitions of the $A_n$, $A_n'$ and $\bm{q}_n$, $\bm{q}_n'$).
    \end{proof}
    
    \paragraph*{The Uniform Compressed Word Problem for Finitary Automaton Groups.}
    \autoref{prop:compressedWPisInPSpace} and \autoref{prop:compressedWPisPSpaceHard} form the two directions for the following theorem.
    \begin{theorem}
      The uniform compressed word problem for finitary automaton groups
      \problem{
        a finitary \GAut $\mathcal{T} = (Q, \Sigma, \delta)$ and\newline
        a straight-line program generating a state sequence $\bm{q} \in Q^{\pm *}$
      }{
        is $\bm{q} = \idGrp$ in $\mathscr{G}(\mathcal{T})$?
      }\noindent
      is \PSpace-complete (under many-one $\LogSpace$-reductions).
      This remains true if we fix a set with five elements as the alphabet of the input automaton.
    \end{theorem}
    
    \paragraph*{Binary Alphabet.}
    We could adapt the above reduction to use the automaton generating Grigorchuk's group from \autoref{ex:grigorchuk} instead of the one for $A_5$. This, however, makes the proof even more technical and there is a direct way to reduce the compressed word problem of Grigorchuk's group, which known to be \PSpace-complete \cite{bartholdi2019groups}, to our problem (although the reduction is less direct).
    
    For this, we extend the notion of finite approximability to SLPs.
    \begin{definition}\label{def:compressiblyfinitelyApproximable}
      A family $\mathcal{F}$ of \GAuta is \emph{compressibly finitely approximable} if the function
      \function
        [the family $\mathcal{F}$ of \GAuta]
        {a \GAut $\mathcal{T} = (Q, \Sigma, \delta) \in \mathcal{F}$ and\newline
         a straight-line program generating a state sequence $\bm{q} \in Q^{\pm*}$}
        {a finitary \GAut $\mathcal{T}' = (Q', \Sigma, \delta')$ and\newline
         a straight-line program generating a state sequence $\bm{q}' \in Q'^{\pm *}$\newline
         with $\bm{q}' \neq \idGrp$ in $\mathscr{G}(\mathcal{T}') \iff \bm{q} \neq \idGrp$ in $\mathscr{G}(\mathcal{T})$
        }\noindent
      is $\LogSpace$-computable.
    \end{definition}

    Again (compare to \autoref{prop:finitelyApproximable}), this notion plays nicely with contracting automata:
    \begin{proposition}\label{prop:compressiblyFinitelyApproximable}
      Let $\mathcal{R}$ be a contracting \GAut. Then any family $\mathcal{F}$ of $\mathcal{R}$-finitary \GAuta is compressibly finitely approximable.
    \end{proposition}
    \begin{proof}
      As a first step, we ensure that all the rules of the input SLP are either of the form $A \to x$ or $A \to xy$ where $x$ and $y$ may be terminal (i.\,e.\ states) or non-terminal symbols. To do this, we successively break up long rules $A \to x_1 x_2 \dots x_\ell$ into $A \to x_1 B$ and $B \to x_2 \dots x_\ell$ with a new non-terminal symbol $B$. This can be done in $\LogSpace$ as we only need a pointer into the input SLP and a counter for the new non-terminals. As there are at most quadratically many such new symbols (every rule may create at most linearly many), this counter can certainly be realized in binary within logarithmic space.
      
      From now on, we may assume that the rules of the input SLP are of the above form (as \LogSpace-computable functions may be composed within \LogSpace). Consider the unique syntax tree for the input SLP (generating $\bm{q}$). On any path from its root to a leaf, every non-terminal symbol may appear only once (as, otherwise, we would be able to generate more than one word), which yields $|V|$ as an upper bound on the depth of the syntax tree where $V$ is the set of non-terminal symbols of the input grammar. Since, by our previous normalization, every node in the syntax tree has at most two children, this shows that $2^{|V|}$ is an upper bound for the number of leaves of the syntax tree. In other words, we have $|\bm{q}| \leq 2^{|V|}$.
      \pagebreak
      
      The rest of the proof is now almost identical to the one for \autoref{prop:finitelyApproximable}: we choose $D = |Q| + A \log |\bm{q}| + B$ where $Q$ comes from the input \GAut $\mathcal{T} = (Q, \Sigma, \delta) \in \mathcal{F}$ (and is, thus, given in unary), $\log |\bm{q}| \leq \log 2^{|V|} = |V|$ is also given in unary and $A$ and $B$ are the constants for the contracting \GAut $\mathcal{R}$ from \autoref{lem:finitarilyContracting}. This means that we may compute the finite approximation $\mathcal{T}'$ of depth $D$ of $\mathcal{T}$ in logarithmic space by \autoref{fct:finiteApproximationLogSpace}.
      
      For the state sequence $\bm{q}'$, we choose the projection of $\bm{q}$ in $\mathcal{T}'$ (see \autoref{def:finiteApproximation}), for which we have already shown
      \[
        \bm{q}' \neq \idGrp \text{ in } \mathscr{G}(\mathcal{T}') \iff
        \bm{q} \neq \idGrp \text{ in } \mathscr{G}(\mathcal{T})
      \]
      in the proof of \autoref{prop:finitelyApproximable}. It remains to describe how we may obtain an SLP generating $\bm{q}'$ from the input one generating $\bm{q}$ within logarithmic space. However, this can easily be done by replacing every terminal symbol $q \in Q$ (or $q^{-1}$) in any rule of the SLP by the new terminal symbol $(q, 0)$ (or $(q, 0)^{-1}$) (which does not even make use of the logarithmic space).
    \end{proof}

    We apply \autoref{prop:compressiblyFinitelyApproximable} only to the singleton family containing the contracting automaton generating Grigorchuk's group (from \autoref{ex:grigorchuk}), which yields our last result:
    \begin{theorem}
      The uniform compressed word problem for finitary automaton groups with binary alphabet
      \problem
        [the binary alphabet $\Sigma = \{ 0, 1 \}$]
        {a finitary \GAut $\mathcal{T} = (Q, \Sigma, \delta)$ and\newline
         a straight-line program generating a state sequence $\bm{q} \in Q^{\pm *}$}
        {is $\bm{q} = \idGrp$ in $\mathscr{G}(\mathcal{T})$?}\noindent
      is \PSpace-complete.
    \end{theorem}
    \begin{proof}
      We only need to show the \PSpace-hard part of the statement by \autoref{prop:compressedWPisInPSpace}.
      Let $\mathcal{G}$ denote the \GAut from \autoref{ex:grigorchuk} generating Grigorchuk's group. The compressed word problem of Grigorchuk's group
      \problem
        [the \GAut $\mathcal{G}= (Q, \{ 0, 1 \}, \delta)$]
        {a straight-line program generating a state sequence $\bm{q} \in Q^{\pm *}$}
        {is $\bm{q} = \idGrp$ in $\mathscr{G}(\mathcal{G})$?}\noindent
      is \PSpace-complete \cite{bartholdi2019groups}. Furthermore, $\mathcal{G}$ is contracting (see \autoref{ex:grigorchuk}) and the singleton family $\{ \mathcal{G} \}$ is trivially $\mathcal{G}$-finitary. Thus, this family is compressibly finitely approximable by \autoref{prop:compressiblyFinitelyApproximable} and this yields that the function from \autoref{def:compressiblyfinitelyApproximable} is a \LogSpace-reduction from the above compressed word problem of Grigorchuk's group to the uniform compressed word problem in the theorem statement.
    \end{proof}
  \end{section}

  \clearpage
  \section*{Acknowledgments}
    The authors would like to thank Armin Weiß for many discussions around the presented topic. The presented results are based on results from the first author's Bachelor thesis, which was advised by the second author (while he was at FMI). This work was mainly produced while the second author was affiliated with the Dipartimento di Matematica of the Politecnico di Milano and funded by the Deutsche Forschungsgemeinschaft (DFG, German Research Foundation) – 492814705. The affiliation listed above is the current affiliation of the second author, partly funded by ERC grant 101097307.

  \bibliographystyle{plain}
  \bibliography{references}
\end{document}